\documentclass[11pt,a4paper]{article}
\usepackage[margin=1in]{geometry}
\usepackage{amsmath,amssymb,amsthm}
\usepackage{graphicx}
\usepackage{float}
\usepackage[colorlinks=true,linkcolor=blue,citecolor=blue]{hyperref}
\usepackage{booktabs}
\usepackage{caption}
\newtheorem{proposition}{Proposition}

\theoremstyle{definition}

\theoremstyle{remark}
\newtheorem{remark}{Remark}

\newcommand{\Var}{\mathrm{Var}}
\newcommand{\Cov}{\mathrm{Cov}}
\newcommand{\Ex}{\mathbb{E}}
\renewcommand{\Pr}{\mathbb{P}}
\newcommand{\Hb}{\mathcal{H}_{\mathrm{blind}}}

\title{Blind directions of linear response and the limits of
finite-perturbation bounds on efficiency fluctuations}
\author{%
Badr Farih\\[2pt]
\normalsize \texttt{badr.farih1@usmba.ac.ma}%
}
\date{}

\begin{document}
\maketitle
\vspace{-2.2em}

\begin{center}\begin{minipage}{0.88\textwidth}
\noindent\textbf{Abstract.}
Thermodynamic uncertainty relations bound the fluctuations of time-antisymmetric
currents. Stochastic efficiency is not one: a ratio of two odd quantities, it is
even under path reversal, and for the conventional ratio $-W/Q_{\rm h}$ no moment
of any order exists. We work with the exergetic ratio $\eta=W/(W+T_0S)$, whose
moments do exist, and ask how much of its variance a finite perturbation of the
dynamics can certify when linear response cannot. The question is specific to
nonlinear observables: we show that for every current, and every other
observable linear in transition counts and residence times, the
multiparameter Cram\'er--Rao bound over rate and site perturbations already
equals the variance, so finite perturbations can add nothing. For $\eta$ they
can. The Cram\'er--Rao bound vanishes identically on a hyperplane of perturbation
directions, while the multi-point bound of Barankin, applied to tilted path
measures, stays strictly positive there and recovers $75\%$ of $\Var(\eta)$ in a
three-state motor model where every linear-response bound is zero to machine
precision. The Gram matrix of tilted Markov-jump path measures has a closed form
as a Feynman--Kac matrix exponential whose potential is the Hellinger integrand
of the jump intensities, and the cost of a perturbation is its $\chi^{2}$
divergence, which no dissipation bound controls. On two models the bound
captures $1.25$--$1.57$ times the best linear-response bound, and the gain
persists, at about $1.3$, when the motor is embedded in networks of up to
sixteen states and a $33$-dimensional perturbation space.
For an Ornstein--Uhlenbeck process the construction is exactly solvable: the
optimal tilt scales as $\mathcal T^{-1/2}$, three test points come within
$3\times10^{-5}$ of the supremum, and the recovered fraction peaks at an intermediate window. Two limits
are established rather than assumed. The gain is a short-window effect that
disappears as the counts become Gaussian. And the bound cannot be turned into
inference: coarse measurements of mean efficiency bound $\chi^{2}$ only from
below, so the one form an experiment could evaluate recovers $11\%$ of the
variance at the shortest window and $0.05\%$ at the longest, with nothing
measurable to say which. The construction is a computational instrument for a
model in hand, not an uncertainty relation.

\vspace{0.6em}
\noindent\textbf{keywords:} stochastic thermodynamics; thermodynamic uncertainty
relations; efficiency fluctuations; Barankin bound; nonlinear response;
Markov jump processes.
\end{minipage}\end{center}

\section{Introduction}\label{sec:intro}
The thermodynamic uncertainty relation states that the precision of a
time-integrated current is limited by the entropy production that pays for
it~\cite{Barato2015,Gingrich2016,Horowitz2017}. It is one of the most useful
results of stochastic thermodynamics~\cite{Seifert2012}, and its usefulness has a boundary that is
rarely made explicit: it is a statement about \emph{currents}, and it is derived
under the assumption that the observable of interest is antisymmetric under time
reversal. Indo and Hasegawa put the point plainly: the antisymmetry
$\phi(\Gamma^{\dagger})=-\phi(\Gamma)$ is ``often implicitly assumed'' and ``has
also become a major constraint, restricting the applicability of the
[relations] to a specific class of physical quantities''~\cite{Indo2025}. Their
own generalisation relaxes it only to $\phi(\Gamma^{\dagger})\phi(\Gamma)\le0$.

The efficiency of a machine is outside that class, for two independent reasons.

The first is symmetry. Efficiency is a ratio of two quantities that are both odd
under path reversal, so it is \emph{even}:
$\eta(\Gamma^{\dagger})=\eta(\Gamma)$. A positive even observable satisfies
$\phi(\Gamma^{\dagger})\phi(\Gamma)=\phi^{2}>0$ and therefore fails even the
relaxed condition of~\cite{Indo2025}. The fluctuation-theorem pairing, applied
to such an observable, returns an identity rather than an inequality.

The second concerns the conventional figure of merit and is more serious. For
$\tilde\eta=-W/Q_{\rm h}$ the stochastic input heat in the denominator
fluctuates through zero, the distribution acquires $\rho(\tilde\eta)\sim
\tilde\eta^{-2}$ tails, and \emph{every} moment diverges~\cite{Holubec2022};
Polettini, Verley and Esposito state that the distribution affords ``no moments
of any order, so that there is no average efficiency and mean-square
error''~\cite{Polettini2015}. This is why the efficiency-fluctuation literature
works with large-deviation functions~\cite{Verley2014} rather than with moments,
and why a variance bound for $\tilde\eta$ has no meaning.

Both obstructions are visible in the shape of the literature. Response bounds
for currents are mature~\cite{Aslyamov2025,Liu2025,Kwon2025}; efficiency is never
the responding observable, and is instead bounded indirectly through its
constituent currents~\cite{Pietzonka2016}.

We take the second obstruction to be the more fundamental, and remove it by
choice of observable. The exergetic (second-law) efficiency
\begin{equation}
\eta=\frac{W}{W+T_0S},
\label{eq:eta}
\end{equation}
work delivered over work delivered plus exergy destroyed, removes it. This is
the standard second-law, or exergetic, efficiency of engineering thermodynamics,
in which the Gouy--Stodola theorem identifies $T_0S$ as the exergy irreversibly
destroyed by the process~\cite{Bejan2016}. We do not propose a new figure of
merit; we apply a classical one in a regime where it has not been used.

The only change is in the denominator, which is the exergy \emph{consumed}
rather than the heat \emph{supplied}. Consumed exergy vanishes only together with
the work delivered, so there is no interior pole. In the jump models of Section~\ref{sec:models} the
consumed exergy is $\Delta\mu$ times an integer net fuel count, so on the event
where it is positive it is bounded below by $\Delta\mu$, and every moment of
$\eta$ exists. Whenever $\Sigma\ge0$ one has $\eta\le1$; $\eta<0$ remains
possible on rare trajectories along which the machine runs backwards. It is also
the ratio that a design or costing exercise wants, since it charges a device with
what it destroyed rather than with what it was handed.

What the energy budget alone determines about the \emph{mean} of \eqref{eq:eta}
is a moment problem of a different kind, treated separately by the present
author~\cite{Farih2026}, and is not our subject. We take \eqref{eq:eta} as given
and ask about its \emph{variance}, with a dynamical model in hand. The whole
paper answers one question: \emph{when an observable is invisible to linear
response, how much of its fluctuation can a finite perturbation of the dynamics
still certify, and at what price?}

\subsection{What this paper does}\label{sec:does}
Perturb the transition rates of a Markov jump process, $w_{ij}\to w_{ij}
e^{u_{ij}/2}$, and let $M_u$ be the resulting Radon--Nikodym derivative on path
space. Cauchy--Schwarz gives, for any observable $O\in L^{2}$,
\begin{equation}
\Var_0(O)\;\ge\;\frac{\Cov_0(O,M_u)^{2}}{\Var_0(M_u)},
\label{eq:hcr}
\end{equation}
which is the Hammersley--Chapman--Robbins bound~\cite{Hammersley1950,Chapman1951};
its $u\to0$ limit along a direction is the Cram\'er--Rao bound, and the supremum
over infinitesimal directions is the multiparameter Cram\'er--Rao bound. Both
ends of this family are in use in stochastic thermodynamics: Hasegawa and Van
Vu use the single-point form with the Pearson divergence~\cite{Hasegawa2019PRE},
and Dechant uses the derivative limit with the Fisher information
matrix~\cite{Dechant2019}. The interpolation between them is not in use.

The interpolation is the Barankin bound~\cite{Barankin1949}. Its content is a
projection: take $m$ perturbations at once and project $O$ onto the span of
$\{M_{u_1},\dots,M_{u_m}\}$ in $L^{2}$. We use it, and the paper is organised
around what becomes possible when one does.

The answer to the question posed above comes in five parts.

\emph{Which observables the question concerns.} Section~\ref{sec:setup} sets up the bound and
Proposition~\ref{prop:currents} (Section~\ref{sec:currents}) shows that for currents it is
empty: every observable linear in transition counts and residence times is
already captured exactly by the multiparameter Cram\'er--Rao bound once site
perturbations are admitted alongside rate perturbations. Finite perturbations
can help only for observables that are nonlinear functionals of the trajectory
record, and efficiency, a ratio of two currents, is among the simplest of
thermodynamic interest. This is why the paper is about $\eta$ and not about a
current.

\emph{The price.} Section~\ref{sec:dissipation} shows that the price of a finite perturbation cannot
be paid in dissipation: no bound $\chi^{2}\le C\langle\Sigma\rangle_0$ with $C$
independent of the perturbation exists. Section~\ref{sec:gram} shows that the price actually
charged, the Gram matrix of $\chi^{2}$-type overlaps, is computable in closed
form.

\emph{What is gained where linear response is blind.} $\chi_1(u)=\partial_\theta
\langle O\rangle_{\theta u}|_{0}$ is \emph{linear} in $u$, so the set on which it
vanishes is a hyperplane, and on that hyperplane every linear-response bound
(Cram\'er--Rao, and every uncertainty relation derived from it) is identically
zero. Section~\ref{sec:blind} shows that the finite-perturbation bound is strictly positive
there; Sections~\ref{sec:design}--\ref{sec:results} show how to choose the perturbations and that the bound
recovers three quarters of $\Var(\eta)$ in the models studied, a fraction and
a gain over linear response that do not degrade on networks of up to sixteen
states (Section~\ref{sec:largeN}).

\emph{What cannot be gained.} Section~\ref{sec:limits} shows that the gain is a short-window
effect that disappears as the counts become Gaussian. Section~\ref{sec:measure} shows that the
bound cannot be converted into a statement about measured data: coarse
measurements bound $\chi^{2}$ only from below, and the upper bound an
experimental version would need is not available from any measurement. We
regard this negative result as being as important as the positive one, and state
it as a result rather than a caveat.

\emph{An exact case.} Section~\ref{sec:diffusion} carries the construction to diffusions and
solves the Ornstein--Uhlenbeck case exactly, which turns the numerical
observations of Sections~\ref{sec:design} and~\ref{sec:results} into formulas.

\subsection{What is not claimed}\label{sec:notclaimed}
The ingredients come from two literatures that do not usually meet, so we state
the boundary at the outset.

The inequality is not new. The multi-point form \eqref{eq:mp} is the Barankin
bound~\cite{Barankin1949}, with Hammersley--Chapman--Robbins as its $m=1$ case
and multiparameter Cram\'er--Rao as its derivative limit; the finite-$m$
Gram-matrix presentation is standard in signal
processing~\cite{McAulay1969,McAulay1971,Todros2010}, and it has recently been
brought into quantum metrology~\cite{Gessner2023}. The observation that
$\chi^{2}$ is the exact optimum of a finite-perturbation response ratio over all
observables is due to Falasco, Esposito and Delvenne~\cite{Falasco2022}.
Finite-perturbation fluctuation-response inequalities in stochastic
thermodynamics are due to Dechant and Sasa~\cite{Dechant2020}, with the
Kullback--Leibler divergence in place of $\chi^{2}$. The structure of the
closed form in Section~\ref{sec:gram} is the classical Hellinger transform of Jacod and
Shiryaev~\cite{Jacod2003,Kabanov1986}, and its point-process case is treated by
Leskel\"a~\cite{Leskela2024}.

What we claim is the combination and five specific things within it: the
evaluation of the Gram matrix for finite-state Markov jump processes
(Section~\ref{sec:gram}); the impossibility of a dissipation-only cost at finite
perturbation (Section~\ref{sec:dissipation}); the blind-direction phenomenon,
together with the observation that linear response is already exact for every
current, which locates the problem in nonlinear observables
(Section~\ref{sec:blind}); matching pursuit as a test-point rule
(Section~\ref{sec:design}); and the two limits of Sections~\ref{sec:limits}
and~\ref{sec:inference}, that the gain is a short-window effect and that coarse
data cannot certify it. Section~\ref{sec:diffusion} adds an exactly solvable
diffusion; its constants are classical once the problem is recognised as a
Gaussian location model, and we claim only their appearance in this setting.

\section{Setup and the multi-point bound}\label{sec:setup}
\subsection{Dynamics, observable, perturbation}\label{sec:dynamics}
Let $\Gamma=\{x_t\}_{t\in[0,\mathcal T]}$ be a trajectory of a time-homogeneous
Markov jump process on a finite state space $V$ with rates $w_{ij}$, in a
nonequilibrium steady state $p_0$. Write $n_e(\Gamma)$ for the number of
transitions along directed edge $e$, $\tau_i(\Gamma)$ for the residence time in
state $i$, and $r_i=\sum_{j\neq i}w_{ij}$ for the escape rate.

Perturb the rates by a vector $u$,
\begin{equation}
w_{ij}\;\longrightarrow\;w_{ij}\,e^{u_{ij}/2},
\label{eq:tilt}
\end{equation}
so that an antisymmetric $u$ changes edge affinities (an ``entropic'' tilt) and a
symmetric one changes traffic (a ``kinetic'' tilt). The latter acts on the
time-symmetric part of the path-space fluctuations~\cite{Maes2006}, the part that
also controls the kinetic uncertainty relation~\cite{DiTerlizzi2019}. By Girsanov the
Radon--Nikodym derivative of the perturbed path measure with respect to the
unperturbed one is
\begin{equation}
M_u(\Gamma)=\exp\!\left[\tfrac12\sum_e u_e n_e(\Gamma)-\sum_{i\in V}\tau_i(\Gamma)\,\Delta r_i(u)\right],
\qquad
\Delta r_i(u)=\sum_{j\neq i}w_{ij}\!\left(e^{u_{ij}/2}-1\right),
\label{eq:M}
\end{equation}
and $\Ex_0[M_u]=1$ identically. This explicit form was given by Seifert and
Speck for specific perturbation parameters~\cite{Seifert2010} and for general
ones by Zheng and Lu~\cite{Zheng2025a}.

The observable is $\eta$ of \eqref{eq:eta}. For the models of Section~\ref{sec:models} the
delivered work and the dissipation are linear in the transition counts, so
$\eta$ is a ratio of two currents.

\begin{remark}[$\eta$ must be a function on all of path space]
\label{rem:convention}
$\eta$ is undefined where the consumed exergy $W+T_0S$ vanishes. Throughout we
use the convention $\eta(\Gamma)=W/(W+T_0S)$ where $W+T_0S>0$ and
$\eta(\Gamma)=0$ elsewhere, so that $\eta$ is a genuine functional on path
space. This is not cosmetic: conditioning on $W+T_0S>0$ instead would make
$\langle\eta\rangle_u$ and $\langle\eta\rangle_0$ expectations of two
\emph{different} maps over two different sets, and \eqref{eq:hcr} would not
apply.

The convention is also not negligible, and we state its weight rather than leave
it to be discovered. Over a short window a small machine frequently fails to
consume any net exergy at all (it idles, or runs backwards), and the resulting
point mass at $\eta=0$ is large (Table~\ref{tab:pointmass} in
Appendix~\ref{sec:suppl}).

At the three-state model's working point, $41\%$ of windows carry no net exergy
consumption. Assigning them $\eta=0$ is a physical statement (a device that
consumed nothing converted nothing), not a technical patch, but it does mean
that $\Var(\eta)$ there is partly the variance of a Bernoulli indicator.
Three things keep the paper's conclusions clear of it. The bound and the
Cram\'er--Rao bound it is compared against are computed for the \emph{same}
functional, so every ratio reported below is unaffected by the convention. For
$\mathcal T\ge4$ the conditional and unconditional variances differ by less
than $10\%$; only at $\mathcal T=2$, where the point mass is largest ($74.7\%$),
do they differ substantially, by about $27\%$ ($3.73$ against
$4.73\times10^{-2}$), and there the first point applies. And the two-channel model at $\mathcal T\in\{1,2\}$ has a point mass of $2.7\%$
and $0.16\%$, negligible by any standard, and reproduces the same conclusions
(Section~\ref{sec:sweep}), so the results do not depend on a model in which the mass is
large.
\end{remark}

\subsection{The bound}\label{sec:bound}
\begin{proposition}[Multi-point projection bound]
\label{prop:mp}
Let $O\in L^{2}(\Omega,\Pr_0)$ and let $\{\Pr_{u_k}\}_{k=1}^{m}$ be path measures
absolutely continuous with respect to $\Pr_0$ with $M_{u_k}\in L^{2}$. Put
$b_k=\Cov_0(O,M_{u_k})$ and $G_{kl}=\Cov_0(M_{u_k},M_{u_l})$. Then
\begin{equation}
\Var_0(O)\;\ge\;b^{\top}G^{-1}b .
\label{eq:mp}
\end{equation}
\end{proposition}

\begin{proof}
Work in $\mathcal{H}=L^{2}(\Omega,\Pr_0)$ with $\langle X,Y\rangle=\Ex_0[XY]$ and
centre: $\Delta O=O-\Ex_0[O]$, $\Delta M_{u_k}=M_{u_k}-1$, using
$\Ex_0[M_{u_k}]=1$. Then $\langle\Delta O,\Delta M_{u_k}\rangle=b_k$ and
$\langle\Delta M_{u_k},\Delta M_{u_l}\rangle=G_{kl}$. Let
$\mathcal{S}=\mathrm{span}\{\Delta M_{u_1},\dots,\Delta M_{u_m}\}$. By the
projection theorem $\Delta O=\Pi_{\mathcal S}\Delta O+\Delta O^{\perp}$ with the
two terms orthogonal, so
$\Var_0(O)=\|\Delta O\|^{2}\ge\|\Pi_{\mathcal S}\Delta O\|^{2}$. Writing
$\Pi_{\mathcal S}\Delta O=\sum_k c_k\Delta M_{u_k}$, orthogonality of the
residual to each $\Delta M_{u_l}$ gives $Gc=b$, hence $c=G^{-1}b$ when $G$ is
non-singular, and $\|\Pi_{\mathcal S}\Delta O\|^{2}=c^{\top}Gc=b^{\top}G^{-1}b$.
\end{proof}

Three remarks fix the status of \eqref{eq:mp} and are used repeatedly below.

\begin{remark}[It is the Barankin bound]
\label{rem:barankin}
Proposition~\ref{prop:mp} is the finite-order Barankin
bound~\cite{Barankin1949,McAulay1969,Todros2010}. Setting $m=1$ recovers
\eqref{eq:hcr}; letting the $u_k$ shrink to zero along $d$ coordinate directions
recovers the multiparameter Cram\'er--Rao bound $v^{\top}I_F^{-1}v$ with
$v_e=\Cov_0(O,s_e)$ and $I_F$ the path-space Fisher information matrix; this
linear-response limit, in multiparameter and multi-observable form, is the
response inequality of Zheng and Lu~\cite[Sec.~V]{Zheng2025a}. Adding
test points can only increase the bound, since the span grows. Barankin's
theorem states that the supremum over test points and over $m$ is the variance
of the locally best unbiased estimator; here that means $\Var_0(O)$ itself, so
the family is in principle exhaustive.
\end{remark}

\begin{remark}[The geometry, and what ``nonlinear'' means]
\label{rem:geometry}
By Cauchy--Schwarz,
\[
\sup_{g\in L^{2}}\ \frac{\Cov_0(O,g)^{2}}{\Var_0(g)}\;=\;\Var_0(O),
\]
attained at $g\propto O$. So
\eqref{eq:hcr} is the projection of $O$ onto the \emph{curve} $\{M_{\theta u}\}$,
Cram\'er--Rao is the projection onto that curve's tangent at the origin (that
is, onto the score alone), and \eqref{eq:mp} is the projection onto the span
of $m$ points on the manifold. In the infinitesimal limit this is the
information geometry of trajectory space defined by the Fisher
metric~\cite{Zheng2025b}; the finite-perturbation bound replaces the tangent
space by chords of the same manifold. What a finite perturbation adds is that $M_u$
carries every power of the score, not only the first. It also follows
that no member of this family can exceed $\Var_0(O)$, which is a useful check on
any numerical implementation.
\end{remark}

\begin{remark}[Test functions, not likelihood ratios]
\label{rem:testfunctions}
Nothing in the proof uses the fact that $M_{u_k}$ is a Radon--Nikodym
derivative. Cauchy--Schwarz gives
$\Cov_0(O,z)^{2}/\Var_0(z)\le\Var_0(O)$ for any $z\in L^{2}$, and the projection
argument goes through for any finite family of square-integrable test
\emph{functions}. This has two consequences used later: the plug-in estimator of
Section~\ref{sec:experiment}, for which $\Ex_0[\hat M_u]\ne1$ at finite sample size, still yields a
valid bound; and the enlarged family of Section~\ref{sec:sitepot}, whose members are not
likelihood ratios of any probability measure, is admissible.
\end{remark}

\section{The cost cannot be dissipation}\label{sec:dissipation}
Before building anything we rule out the construction one would most like to
have. A thermodynamic uncertainty relation bounds fluctuations by
$\langle\Sigma\rangle$; one might hope that the denominator of \eqref{eq:hcr} or
\eqref{eq:mp} could be replaced by dissipation in the same way. At infinitesimal
perturbation it can, and that is the ordinary relation. At finite perturbation
it cannot, and the obstruction is not technical.

\begin{remark}[No tilt-independent dissipation bound]
\label{rem:impossible}
There is no constant $C<\infty$, independent of $u$, with
$\chi^{2}(\Pr_u\Vert\Pr_0)\le C\,\langle\Sigma\rangle_0$. Fix an edge
$e=(i\to j)$ with $w_e>0$ and tilt only it, $u_e=\lambda>0$, so that its rate
becomes $w_ee^{\lambda/2}$ and the escape rate of $i$ becomes
$r_i+w_e(e^{\lambda/2}-1)$. Let $B$ be the event that $x_0=i$ and that the only
jump in $[0,\mathcal T]$ is a single jump $i\to j$ along $e$. On $B$, with jump
time $s$,
$M_u=e^{\lambda/2}\exp\!\big(-w_e(e^{\lambda/2}-1)s\big)$, and since
$\chi^{2}(\Pr_u\Vert\Pr_0)=\Ex_0[M_u^{2}]-1\ge\Ex_0[M_u^{2}\mathbf 1_B]-1$,
\begin{equation}
\chi^{2}(\Pr_u\Vert\Pr_0)+1\;\ge\;p_{0,i}\,w_e e^{\lambda}e^{-r_j\mathcal T}
\,\frac{1-e^{-\mathcal T(2w_ee^{\lambda/2}+r_i-2w_e)}}{2w_ee^{\lambda/2}+r_i-2w_e}
\;\sim\;\tfrac12\,p_{0,i}\,e^{-r_j\mathcal T}\,e^{\lambda/2}
\qquad(\lambda\to\infty).
\label{eq:impossible}
\end{equation}
The right-hand side diverges, while $\langle\Sigma\rangle_0$ is a property of
the unperturbed steady state and does not depend on $u$ at all. (A single test
observable does not suffice here: the response $\langle n_e\rangle_u-\langle
n_e\rangle_0$ of the edge count stays bounded, because a fast edge is quickly
starved by the rest of the network; what diverges is the rare-event mass that
$M_u^2$ puts on trajectories using $e$ once.)
\end{remark}

Stated this way the obstruction is close to definitional, which is why we
present it as a remark. Its value is that it closes a direction: entropy
production controls the $\theta\to0$ limit and nothing beyond it, so the
finite-perturbation regime lies outside the reach of any uncertainty relation
whose only currency is dissipation. What replaces dissipation as the cost is
$\chi^{2}$ itself; Section~\ref{sec:gram} shows that this is a computable object and
Section~\ref{sec:experiment} that it is a measurable one.

The divergence is fast. Along the optimal direction of the three-state model of
Section~\ref{sec:threestate}, at fixed $\langle\Sigma\rangle_0=2.38$:

\begin{table}[H]
\centering
\caption{$\chi^{2}$ at fixed dissipation, along the optimal perturbation
direction of the three-state model. $\langle\Sigma\rangle_0$ does not depend on
$|u|$; $\chi^{2}$ is unbounded in it}
\label{tab:impossible}
\small
\begin{tabular}{lccccccc}
\toprule
$|u|$ & 0.042 & 0.083 & 0.167 & 0.333 & 0.500 & 0.667 & 1.000\\
$10^{3}\,\chi^{2}/\langle\Sigma\rangle_0$ & 0.22 & 0.87 & 3.4 & 13.4 & 29.7 & 53.0 & 124\\
\bottomrule
\end{tabular}
\end{table}

\section{The Gram matrix in closed form}\label{sec:gram}
The bound \eqref{eq:mp} needs $G$, and for path measures $G$ is an expectation
over trajectories. It has a closed form.

\begin{proposition}[Gram matrix of tilted jump-process path measures]
\label{prop:fk}
For the tilt \eqref{eq:tilt} of a finite-state Markov jump process over
$[0,\mathcal T]$,
\begin{equation}
G_{kl}+1\;=\;\big\langle\mathbf{1}\big\vert\,
\exp\!\Big(\mathcal T\big[\mathcal L_{u_k+u_l}+\mathrm{diag}\,g_{kl}\big]\Big)
\big\vert p_0\big\rangle,
\qquad
g_{kl,i}=\sum_{j\neq i}w_{ij}\big(e^{u_{k,ij}/2}-1\big)\big(e^{u_{l,ij}/2}-1\big),
\label{eq:fk}
\end{equation}
where $\mathcal L_{u_k+u_l}$ is the ordinary generator of the process tilted at
$u_k+u_l$. On the diagonal, $g_{kk,i}=\sum_j w_{ij}(e^{u_{ij}/2}-1)^{2}\ge0$
and $G_{kk}=\chi^{2}(\Pr_{u_k}\Vert\Pr_0)$.
\end{proposition}

\begin{proof}
Multiplying two copies of \eqref{eq:M},
\[
M_{u_k}M_{u_l}=\exp\!\left[\sum_e n_e\frac{u_{k,e}+u_{l,e}}{2}
-\sum_i\tau_i\big(\Delta r_i(u_k)+\Delta r_i(u_l)\big)\right].
\]
The jump weight $e^{(u_k+u_l)_e/2}$ is the off-diagonal part of
$\mathcal L_{u_k+u_l}$, whose diagonal carries $-r_i-\Delta r_i(u_k+u_l)$ rather
than the required $-r_i-\Delta r_i(u_k)-\Delta r_i(u_l)$. The mismatch is
\begin{align*}
\Delta r_i(u_k{+}u_l)-\Delta r_i(u_k)-\Delta r_i(u_l)
&=\sum_{j}w_{ij}\Big[\big(e^{a+b}-1\big)-\big(e^{a}-1\big)-\big(e^{b}-1\big)\Big]\\
&=\sum_j w_{ij}\big(e^{a}-1\big)\big(e^{b}-1\big)\;=\;g_{kl,i},
\end{align*}
with $a=u_{k,ij}/2$, $b=u_{l,ij}/2$. Hence
$M_{u_k}M_{u_l}=M_{u_k+u_l}\exp\!\big(\int_0^{\mathcal T}g_{kl,x_t}dt\big)$ and
$\Ex_0[M_{u_k}M_{u_l}]=\Ex_{u_k+u_l}\big[e^{\int g_{kl}}\big]$, which is
\eqref{eq:fk} by Feynman--Kac.
\end{proof}

Two immediate consequences. Since $g_{kk}\ge0$, $\chi^{2}\ge0$ automatically, as
a divergence should be. And $\chi^{2}+1\le e^{\mathcal T\max_i g_i}$, an
envelope that is worthless in practice: at $\theta=1$ in the model of
Section~\ref{sec:threestate} it gives $2.7\times10^{5}$ against a true $\chi^{2}$ of $0.386$. We
record it only to note that the obvious state-independent bound leads nowhere.

\subsection{Every order at once, and what is classical}\label{sec:alpha}
The same computation gives all moments of $M_u$, which is a better-anchored
statement than the $\chi^{2}$ case alone:
\begin{equation}
\Ex_0\!\left[M_u^{\alpha}\right]
=\big\langle\mathbf{1}\big\vert e^{\,\mathcal T[\mathcal L_{\alpha u}+\mathrm{diag}\,c_\alpha(u)]}\big\vert p_0\big\rangle,
\qquad
c_{\alpha,i}(u)=\sum_j w_{ij}\Big[e^{\alpha u_{ij}/2}-\alpha e^{u_{ij}/2}+\alpha-1\Big],
\label{eq:alpha}
\end{equation}
with $c_2=g$ recovering Proposition~\ref{prop:fk}.

The potential in \eqref{eq:alpha} is a classical object. With reference
intensity $\lambda'=w$ and $\lambda=we^{u}$, the intensity of the doubled
tilt $2u$ (recall that a tilt $u$ multiplies rates by $e^{u/2}$), the
Tsallis/Hellinger integrand $\alpha\lambda+(1-\alpha)\lambda'
-\lambda^{\alpha}\lambda'^{1-\alpha}$ equals $\tfrac12 w(e^{u/2}-1)^{2}
=\tfrac12(\sqrt\lambda-\sqrt{\lambda'})^{2}$ at $\alpha=\tfrac12$. So the
diagonal potential $g$ is twice the density of the order-$\tfrac12$ Hellinger
process of Jacod and Shiryaev~\cite{Jacod2003,Kabanov1986}, equivalently the
R\'enyi-$\tfrac12$ divergence rate of the jump intensities; Leskel\"a gives the
point-process case explicitly~\cite{Leskela2024}. (The classical definition runs
over $0<\alpha<1$, where the integrand is non-negative by weighted
arithmetic--geometric mean; that is why $\alpha=2$ sits outside it and why $g$
enters \eqref{eq:fk} with the sign it does.)

What we have not found displayed is the evaluation for a finite-state
continuous-time chain. The discrete-time R\'enyi divergence rate between finite
Markov sources is a classical spectral-radius formula~\cite{Rached2001}; discrete
time has no waiting-time contribution, so the additive potential is exactly the
continuous-time feature that formula does not have. Falasco, Esposito and
Delvenne define $\chi^{2}$ on trajectory space and prove it is the exact optimum
of a finite-perturbation response ratio, but do not evaluate
it~\cite{Falasco2022}; \eqref{eq:fk} is the missing evaluation.

Equation \eqref{eq:fk} agrees with direct simulation of the three-state model to
relative error $\le10^{-4}$ on the diagonal and $7\times10^{-3}$ off it, and
with the two-channel model of Section~\ref{sec:twochannel}, where the state
space has one element and the matrix exponential collapses to a scalar, to
$2.3\%$ (Appendix~\ref{sec:checks}).

\section{Blind directions}\label{sec:blind}
\subsection{The hyperplane on which linear response is silent}\label{sec:hyperplane}
The first-order susceptibility of $O$ along a direction $u$ is
\begin{equation}
\chi_1(u)=\left.\frac{d\langle O\rangle_{\theta u}}{d\theta}\right|_{\theta=0}
=\sum_e u_e\,\Cov_0(O,s_e)=u^{\top}v,
\qquad v_e=\Cov_0(O,s_e),
\label{eq:chi1}
\end{equation}
which is \emph{linear} in $u$. The multiparameter Cram\'er--Rao bound in
direction $u$ is $\mathcal B_{\rm CR}(u)=(u^{\top}v)^{2}/(u^{\top}I_Fu)$. Hence

\begin{proposition}[Blind directions]
\label{prop:blind}
Let $O\in L^{2}(\Omega,\Pr_0)$ with $\Var_0(O)>0$ and $v\neq0$, and let
$\chi_2(u)=\partial^{2}_{\theta}\langle O\rangle_{\theta u}|_{\theta=0}$, a
quadratic form in $u$.
\begin{enumerate}
\item $\mathcal B_{\rm CR}$ vanishes identically on the codimension-one
hyperplane $\Hb=\{u:u^{\top}v=0\}$.
\item For any $u\in\Hb$ \emph{with $\chi_2(u)\neq0$},
\begin{equation}
\lim_{\theta\to0}\frac{1}{\theta^{2}}\,
\frac{\big[\langle O\rangle_{\theta u}-\langle O\rangle_0\big]^{2}}
{\chi^{2}(\Pr_{\theta u}\Vert\Pr_0)}
=\frac{\chi_2(u)^{2}}{4\,\mathcal I_F(u)}>0,
\label{eq:blindlimit}
\end{equation}
so the single-point bound \eqref{eq:hcr} is strictly positive at every
sufficiently small $\theta\neq0$.
\item Such $u$ exist whenever $\chi_2$ does not vanish identically on $\Hb$,
which is generic: the zero set of a non-zero quadratic form is a measure-zero
subvariety of the hyperplane.
\end{enumerate}
\end{proposition}

\begin{proof}
(1) is immediate from \eqref{eq:chi1}. For (2), expand numerator and
denominator. Since $\chi_1(u)=0$,
$\langle O\rangle_{\theta u}-\langle O\rangle_0=\tfrac{\theta^{2}}{2}\chi_2(u)
+\mathcal O(\theta^{3})$, whose square is
$\tfrac{\theta^{4}}{4}\chi_2(u)^{2}+\mathcal O(\theta^{5})$. Meanwhile
$M_{\theta u}=1+\theta s_1(u)+\mathcal O(\theta^{2})$ gives
$\chi^{2}(\Pr_{\theta u}\Vert\Pr_0)=\Ex_0[(M_{\theta u}-1)^{2}]
=\theta^{2}\mathcal I_F(u)+\mathcal O(\theta^{3})$ with
$\mathcal I_F(u)=u^{\top}I_Fu>0$. The ratio is
$\theta^{2}\chi_2(u)^{2}/(4\mathcal I_F(u))+\mathcal O(\theta^{3})$. (3) is the
statement that a non-zero polynomial does not vanish on a set of full measure.
\end{proof}

The condition in (2) belongs on the direction and not on the observable:
$\chi_2$ is a quadratic form, and it may vanish for particular $u\in\Hb$ while
being non-zero for others. Part (3) is what prevents the proposition being
vacuous.

\emph{Why blindness is generic rather than exceptional.} The mechanism is not
delicate. $\chi_1$ is a \emph{linear} functional of the perturbation, so its
kernel is a hyperplane whatever the model: blindness is a codimension-one
condition, present in every system with more than one independently perturbable
edge, and it is the \emph{absence} of blind directions that would need
explaining. What is model-dependent is only which directions they are. In the
three-state model, $\eta$ responds to all four edges with the same sign
($v=(0.021,\,0.128,\,0.065,\,0.018)$), so a blind combination is one in which
the response to tilting the fuel edge is cancelled by an opposite tilt of the
working stroke.

We therefore claim nothing exotic for the existence of $\Hb$. That a first
variation vanishes on the kernel of a linear functional, and that the leading
surviving term is then of second order, is elementary; the content of
Proposition~\ref{prop:blind} is the identification of the limiting ratio
\eqref{eq:blindlimit}. The geometry leaves one quantitative question open:
\emph{how much} of the variance the finite-perturbation bound recovers on $\Hb$.
The second-order term could be arbitrarily small compared with $\Var(\eta)$, and
the construction would then be correct but useless. That it reaches $75\%$ on
the three-state model and $78\%$ on the two-channel model is a measured fact
about those models with no a priori guarantee behind it, and it is the only
reason the hyperplane is worth computing on at all.

The cancellation cannot be arranged on a single edge, and the reason is a
theorem rather than an accident.

\begin{remark}[A blind direction cannot be a single edge]
\label{rem:multiedge}
Bao and Liang prove an identity relating nonlinear to linear response for
single-edge and single-vertex perturbations~\cite{Bao2024}, under which a
vanishing $\chi_1$ forces a vanishing $\chi_2$. On such perturbations
Proposition~\ref{prop:blind} would be empty: blindness and uselessness would
coincide, and there would be nothing to report. It is not empty, because $\Hb$
is a hyperplane through the origin and its non-trivial elements are necessarily
\emph{combinations} of edges. On a combination the per-edge first-order responses
cancel against one another by construction, while the second-order terms (a
quadratic form, not a linear one) do not, and no identity forces them to. In the
three-state model the blind direction inside the first two edges is
$u\propto(1,-0.16175,0,0)$, and the measured susceptibility there is
$\chi_1=2\times10^{-13}$ against a Cram\'er--Rao bound of $8\times10^{-26}$.

So the phenomenon requires a network. It is invisible in any model perturbed one
edge at a time, which may be why it has not been reported.
\end{remark}

\subsection{The scale of the effect, and why one point is not enough}\label{sec:scale}
Equation \eqref{eq:blindlimit} says the bound vanishes as $\theta^{2}$ and is
positive thereafter; it must therefore rise, and then fall once $\chi^{2}$ grows
faster than the squared response, so an interior optimum is guaranteed. Fig.~\ref{fig:blind}(a)
shows both behaviours. Along a responsive direction the bound is flat at its
Cram\'er--Rao value for small $\theta$ and falls away; along the blind direction
it rises as $\theta^{2}$ from zero, peaks near $\theta=2.5$, and falls.

At the peak the single-point blind bound is $0.297\%$ of $\Var(\eta)$. That is
strictly positive where Cram\'er--Rao is $8\times10^{-26}$, i.e. zero to machine
precision, but small. Section~\ref{sec:design} is what turns it into a useful number.

\subsection{Why the observable must be nonlinear}\label{sec:currents}

Blind directions exist for any observable with $v\neq0$, currents included. What
distinguishes $\eta$ is what happens \emph{off} the hyperplane, and the answer
is sharp once the family of perturbations is enlarged by site potentials, which
weight the residence time in state $i$ by $e^{-V_i\tau_i}$ (they are used as test
functions in Section~\ref{sec:sitepot}; here only their infinitesimal scores
$s_{V,i}=-(\tau_i-\langle\tau_i\rangle_0)$ are needed).

\begin{proposition}[Linear response is exact for currents]
\label{prop:currents}
Let $\mathcal S$ be the span in $L^{2}(\Pr_0)$ of the centred scores of all edge
tilts \eqref{eq:tilt} and all site potentials. Then
$\mathcal S=\mathrm{span}\{n_e-\langle n_e\rangle_0,\ \tau_i-\langle\tau_i\rangle_0\}$.
Consequently, for every observable of the form
$O=\sum_e c_en_e+\sum_i d_i\tau_i$, and in particular for every time-integrated
current, the multiparameter Cram\'er--Rao bound over this family equals
$\Var_0(O)$. If the tilts are restricted to antisymmetric ones, one parameter
per bond as in the models below, the same holds with $n_e$ the net count on each
bond, which still covers every current.
\end{proposition}

\begin{proof}
Differentiating \eqref{eq:M} at $u=0$ gives the edge score
$s_e=\tfrac12(n_e-\langle n_e\rangle_0)-\sum_i\partial_{u_e}\Delta r_i(0)\,
(\tau_i-\langle\tau_i\rangle_0)$. Adding the site scores with the coefficients
$\partial_{u_e}\Delta r_i(0)$ isolates $\tfrac12(n_e-\langle n_e\rangle_0)$, so
every centred count and every centred residence time lies in $\mathcal S$, and
conversely every score is a combination of them. The multiparameter
Cram\'er--Rao bound is $\|\Pi_{\mathcal S}\Delta O\|^{2}$ (Remark~\ref{rem:geometry}),
which equals $\|\Delta O\|^{2}=\Var_0(O)$ when $\Delta O\in\mathcal S$.
\end{proof}

So for a current there is nothing for a finite perturbation to add: linear
response over the enlarged family already saturates the variance. On the
three-state model of Section~\ref{sec:threestate} ($2\times10^{5}$ trajectories)
the net count of the work stroke has Cram\'er--Rao $88.2\%$ of its variance over
edge tilts alone and $100.0000\%$ over the enlarged family; the fuel count has
$90.7\%$ and $100.0000\%$, and the residence time in state $0$ has $98.4\%$ and
$100.0000\%$. For $\eta$ the same two numbers are $62.7\%$ and $69.0\%$, and
for the square of the work-stroke count, a nonlinear control, $51.1\%$ and
$53.4\%$. The gap left by $\eta$ is exactly its
nonlinearity in the record $(n,\tau)$, and it is this gap that the rest of the
paper measures. Efficiency is therefore not merely an example: a ratio of
currents is among the simplest observables of thermodynamic interest for which
the question of this paper is not empty.

\section{Choosing the test points}\label{sec:design}
\subsection{Why the choice is the whole problem}\label{sec:whyselect}
The supremum in Barankin's theorem is over test points as well as coefficients,
and neither the estimation literature nor the physics literature offers a
placement theory. Practice in signal processing is domain heuristics plus
empirical saturation: McAulay and Hofstetter place points at matched-filter
sidelobe peaks and report no improvement past four~\cite{McAulay1971}. Marzetta
recasts the supremum over coefficients as an unconstrained quadratic
problem~\cite{Marzetta1997}. We found no convergence rate in $m$ and no greedy
or submodularity result.

Two observations make the problem tractable. First, adding a point raises
\eqref{eq:mp} by exactly the squared correlation of the \emph{residual} with the
new direction after orthogonalisation against the existing span. Second, $M$
depends on an amplitude and a direction only through their product, so
``optimise amplitude and direction'' is a single unconstrained optimisation over
$u\in\mathbb R^{d}$.

\subsection{Three rules}\label{sec:rules}
\emph{Matching pursuit.} Maintain an orthonormal basis $Q$ of the current centred
span, form the residual $r=\Delta O-QQ^{\top}\Delta O$, and take
\begin{equation}
u_{m+1}=\arg\max_{u}\ \frac{\langle r,x_\perp\rangle^{2}}{\|x_\perp\|^{2}},
\qquad x_\perp=x-QQ^{\top}x,\quad x=M_u-\Ex_0[M_u],
\label{eq:pursuit}
\end{equation}
a continuous optimisation at each step rather than a search over a list.

\emph{Joint optimisation.} Move all $m$ points together by quasi-Newton descent.
With $\alpha=(G+\rho I)^{-1}b$,
\begin{equation}
\frac{\partial}{\partial u_k}\big(b^{\top}(G+\rho I)^{-1}b\big)
=2\alpha_k\left[\frac{\partial b_k}{\partial u_k}
-\sum_l \alpha_l\frac{\partial G_{kl}}{\partial u_k}\right],
\label{eq:grad}
\end{equation}
because $b_k$ depends only on $u_k$ while $G_{lm}$ depends on $u_l$ and $u_m$,
and the two contributions coincide by symmetry of $G$. Every ingredient is an
empirical covariance: with $S_k$ the per-trajectory score of point $k$,
$\partial_{u_k}b_k=\Cov_0(O,M_kS_k)$ and
$\partial_{u_k}G_{kl}=\Cov_0(M_kS_k,M_l)$. The ridge $\rho$ must be an
\emph{absolute} constant fixed once for the optimisation; a ridge proportional to
$\mathrm{tr}\,G$ depends on the parameters and \eqref{eq:grad} then omits
$-(d\rho)\|\alpha\|^{2}$, an error worth $22\%$ of the gradient norm and $12$
degrees of direction in our tests. With a fixed ridge, \eqref{eq:grad} agrees
with central differences to $6\times10^{-9}$ at zero angle.

\emph{Dictionary.} Skip selection: take a large fixed dictionary of candidates
and solve one regularised problem over all of them. Marzetta showed that the
Barankin supremum can be written as the unconstrained maximisation of a
\emph{concave quadratic} functional rather than of a ratio of
quadratics~\cite{Marzetta1997}. In his formulation the test points range over a
continuum $\beta$ with an unknown weight \emph{function} $f$, the kernel is the
uncentred $k_A(B,C)=\Ex_A[M_BM_C]$, and the bound is the maximum over $f$ of
\begin{equation}
\int_{\beta}\!\big[(B-A)f^{\top}(B)+f(B)(B-A)^{\top}\big]dB
\;-\;\int_{\beta}\!\int_{\beta} k_A(B,C)\,f(B)f^{\top}(C)\,dB\,dC ,
\label{eq:marzetta}
\end{equation}
closed by an integral equation for the maximiser. Restricted to a finite
dictionary and to a scalar observable this is
$\max_c\,[2c^{\top}b-c^{\top}Gc]=b^{\top}G^{-1}b$ at $c=G^{-1}b$, which is what
we solve, with the ridge chosen on a held-out split.

Marzetta's motivation was simplicity and separability rather than the
ill-conditioning that makes the concave form indispensable here, and he noted
that the continuum problem had not yet been attacked with modern numerical
tools~\cite{Marzetta1997}. Our dictionary rule is a finite-dimensional
discretisation of that programme; the integral-equation version remains untried.

\subsection{None of them dominates}\label{sec:nodominate}
Table~\ref{tab:select} gives the three rules on the three-state model at $m=12$.

\begin{table}[H]
\centering
\caption{Test-point selection on the three-state model, extended family,
$m=12$, as a percentage of $\Var(\eta)$. Joint optimisation from a fresh start
wins in the full space; greedy pursuit wins on the blind hyperplane. Refining a
greedy design jointly does not work: the training objective moves from $79.7\%$
to $79.8\%$ in 300 iterations, i.e. greedy sits in a local optimum of the
span-selection problem that the optimiser cannot leave}
\label{tab:select}
\begin{tabular}{lcc}
\toprule
 & full space & blind hyperplane\\
\midrule
greedy matching pursuit & 80.2 & \textbf{72.3}\\
greedy design, then joint & 78.8 & 61.9\\
Hessian-aligned start, then joint & \textbf{86.4} & 64.3\\
random start, then joint & 83.9 & 65.1\\
\bottomrule
\end{tabular}
\end{table}

Since each design yields a valid bound whatever its provenance, the natural
procedure is to run all three and report the largest; across the sweeps of
Section~\ref{sec:results} each of the three wins at least once. Initialisation matters for the
joint rule, which may seem surprising: the projection onto a \emph{fixed} span is
convex and initialisation-independent, but choosing the span is the outer
problem, and it is not convex.

\subsection{A numerical trap}\label{sec:trap}
$G$ is severely ill-conditioned, so the coefficients must come from a truncated
eigendecomposition. The columns must be standardised first. Infinitesimal basis
tilts have $|u|\sim10^{-5}$ and hence centred columns of norm $\sim10^{-5}$, so a
\emph{relative} eigenvalue cutoff deletes them outright: in our first
implementation this dropped the effective rank from 6 to 2 and made the reported
bound \emph{fall} from $63\%$ to $30\%$ as points were added. A bound that
decreases when test points are added is always an estimator failure and never
the theorem, since the span can only grow.

\subsection{Scaling with the size of the network}\label{sec:scaling}
Both models here are small, and the obvious question is what happens when
$|V|$ grows: what the Gram matrix costs to build, and whether its conditioning
makes the solve impossible without heavy regularisation. The second question is
the sharper one, and it can be answered exactly, with no Monte Carlo, because
$G$ depends only on the generator and the test points and not on the observable.
That is what the closed form buys.

We take a driven ring of $N$ states: bond $i$ carries $i\to i+1$ at rate
$k\,e^{+\mathcal A/2N}$ and $i+1\to i$ at $k\,e^{-\mathcal A/2N}$, so the cycle
affinity is $\mathcal A$ and the tilt space has dimension $N$. Throughout,
$\mathcal A=4$, $k=1$, $\mathcal T=4$, and the test points are random directions
of fixed norm $0.35$.

\emph{Cost.} Building $G$ requires $m(m+1)/2$ matrix exponentials of size
$N\times N$, hence $\mathcal O(m^{2}N^{3})$ in principle. In practice, for
$m=8$, the whole matrix takes $2$~ms at $N=3$ and $11$~ms at $N=48$: a
sixteenfold increase in $N$ for a fivefold increase in time, because dense
$\exp$ on matrices this small is dominated by overhead rather than by the cubic
term. Cost is not the obstacle at any size we can simulate.

\emph{Conditioning.} Table~\ref{tab:scaling} reports $\kappa(G)$ against the
number of test points at fixed $N$, and against $N$ at fixed $m$.

\begin{table}[H]
\centering
\caption{Conditioning of the Gram matrix on a driven $N$-ring. Left: against
$m$ at $N=12$. Right: against $N$ at $m=12$. Ranks are at relative tolerance
$10^{-10}$ and are full in every case}
\label{tab:scaling}
\begin{tabular}{rr@{\qquad\qquad}rr}
\toprule
$m$ & $\kappa(G)$ & $N$ & $\kappa(G)$ \\
\midrule
 2 & $1.2\times10^{0}$  &  3 & $2.8\times10^{5}$ \\
 4 & $2.9\times10^{0}$  &  6 & $9.1\times10^{2}$ \\
 8 & $2.1\times10^{1}$  & 12 & $1.1\times10^{2}$ \\
12 & $4.9\times10^{2}$  & 24 & $1.2\times10^{1}$ \\
16 & $6.7\times10^{2}$  & 48 & $5.9\times10^{0}$ \\
24 & $2.8\times10^{3}$  &    &                   \\
\bottomrule
\end{tabular}
\end{table}

The right-hand column runs the opposite way to the worry. Conditioning
\emph{improves} as the network grows, by nearly five orders of magnitude between
$N=3$ and $N=48$. What degrades $G$ is not the size of the state space but the
ratio of the number of test points to the dimension of the tilt space. Twelve
random directions in a three-dimensional tilt space are severely redundant: they
remain linearly independent as path functionals, since $u\mapsto M_u$ is
nonlinear, but only weakly, and $\kappa$ measures exactly that. Twelve random
directions in forty-eight dimensions are nearly orthogonal. The ill-conditioning
documented in Section~\ref{sec:trap} is therefore a small-network artefact, and the
construction is better behaved on the large networks where it has not yet been
tried than on the ones where it has.

The same improvement appears on random 3-regular graphs and on disordered
$L\times L$ grids with no symmetry (Appendix~\ref{sec:topo}): at equal
dimension $d$ of the tilt space the three topologies agree to within a factor
of about two and a half, so the improvement is set by $d$ and not by the
symmetry of the ring.

Two caveats. Column standardisation, which Section~\ref{sec:trap} shows is not optional,
changes $\kappa$ by under one per cent here, because these test points already
share a common norm; it matters when the design mixes finite tilts with
infinitesimal basis directions, which is the case that produced the failure
recorded in Section~\ref{sec:trap}. And this settles conditioning and cost
only; whether the \emph{gain} survives is a separate question, answered next.

\subsection{Does the gain survive on larger networks?}\label{sec:largeN}

To answer it the network needs an energetics, and the cleanest one is the
three-state motor of Section~\ref{sec:threestate} itself, stretched. We keep its
fuel edge $0\to1$ (force $\Delta\mu$), its work stroke $1\to2$ (force $-w$) and
its futile edge, and replace the single return bond $2\to0$ by a chain of $N-2$
neutral bonds $2\to3\to\dots\to N-1\to0$, all with a common rate $k_n$. The
return is then diffusive and slow, so $k_n$ is calibrated to hold the mean
productive cycle current at its three-state value; $N=3$, $k_n=1$ is the model
of Section~\ref{sec:threestate} exactly. The energetics, the observable $\eta$
and the window $\mathcal T=8$ are unchanged, while the tilt space grows to $N+1$
edge parameters and $2N+1$ parameters with site potentials. For each $N$ we
simulate $2\times10^{5}$ trajectories and score, by the split-sample protocol of
Appendix~\ref{sec:protocol}, a design made of the $2N+1$ infinitesimal basis
tilts plus six finite points chosen by matching pursuit, and separately a
dictionary of $200$ random finite tilts; the table reports the better of the
two. The joint rule, which needs the Hessian of the response and is the most
expensive, is not run, so the numbers are conservative: at $N=3$ the best of all
three rules gives $86.0\%$ (Table~\ref{tab:chain}) against $82.7\%$ here.

\begin{table}[H]
\centering
\caption{The three-state motor stretched to an $N$-state ring, $\mathcal T=8$,
$2\times10^{5}$ trajectories per $N$, split-sample protocol. $d$ is the number
of parameters with site potentials. CR: multiparameter Cram\'er--Rao over edge
tilts (the best linear-response bound) and over edge tilts plus site
potentials. ``best'' is the largest bound over all designs found; ``gain'' is
its ratio to CR over edge tilts; ``blind'' is the bound on the blind hyperplane,
where both Cram\'er--Rao bounds vanish}
\label{tab:largeN}
\begin{tabular}{rrrrrrr}
\toprule
$N$ & $d$ & CR, edges & CR, edges + sites & best & gain & blind\\
\midrule
 3 &  7 & 62.6 & 68.9 & 82.7 & $1.32\times$ & 71.1\\
 6 & 13 & 62.1 & 70.1 & 82.1 & $1.32\times$ & 75.6\\
10 & 21 & 60.8 & 70.7 & 77.5 & $1.27\times$ & 77.5\\
16 & 33 & 59.4 & 71.0 & 77.0 & $1.30\times$ & 77.0\\
\bottomrule
\end{tabular}
\end{table}

The gain survives. Over a nearly fivefold increase in the dimension of the
perturbation space it stays between $1.27$ and $1.32$, and on the blind
hyperplane, where linear response gives exactly zero at every $N$ (the largest
residual $|v\cdot p|$ over the designs is $4\times10^{-16}$), the bound stays
between $71\%$ and $78\%$ of $\Var(\eta)$. What does change is the search. At
$N\le6$ the best design is found by pursuit in the full space; at $N\ge10$
pursuit in the full space falls behind (it scores $75.9\%$ at $N=10$ and
$72.3\%$ at $N=16$, the latter from the dictionary rule), and the best designs
come from pursuit on the blind hyperplane, which is a valid full-space design
too. Pursuit uses a diagonal preconditioning that measures each parameter in
units of the standard deviation of its score, since the fast neutral bonds
otherwise dominate $\chi^{2}$; a first run without it, and with a
derivative-free optimiser, gave only $28\%$ on the blind hyperplane at $N=16$. What limits the
bound at larger $N$ is therefore, as far as we can tell, the optimiser rather
than the phenomenon, which is the conclusion of Section~\ref{sec:nodominate} at
larger scale. The script and
its output are in the replication package~\cite{FarihCode2026}; $N=16$ takes
about $45$ CPU-minutes.

\subsection{How much the design matters}\label{sec:placementsum}

Since no rule dominates, one may ask how much the reported numbers depend on the
design at all. Appendix~\ref{sec:placement} measures it on the three-state
model over ten independent random pools of test points, with the trajectory
ensemble held fixed. The design matters at the scale of a few percentage points
and not tens (standard deviation about $1.2$ points across pools at every $m$),
and the edge-tilt family saturates: a fit to the climb extrapolates to about
$74\%$ of $\Var(\eta)$, which enlarging the family with site potentials
(Section~\ref{sec:sitepot}) raises to $84.0\%$.

\section{Models}\label{sec:models}
\subsection{A three-state motor with a futile cycle}\label{sec:threestate}
Three states, four edges, two cycles sharing the state space: a productive cycle
$0\to1\to2\to0$ consuming one fuel unit and delivering work $w$, and a futile
cycle through a parallel $0\leftrightarrow1$ edge consuming fuel and delivering
none. Rates $w_{ij}=k_e e^{\pm F_e/2}$ with forces $F=(\Delta\mu,-w,0,0)$ and
$k=(1,1,1,0.05)$; baseline $\Delta\mu=4$, $w=2$, $\mathcal T=8$, giving
$\langle\eta\rangle=0.193$, $\Var(\eta)=5.61\times10^{-2}$,
$\langle\Sigma\rangle=2.38$. Transition counts are correlated through the shared
occupancy and the stationary distribution shifts under perturbation.

\subsection{Enlarging the family: site potentials}\label{sec:sitepot}
In \eqref{eq:M} the dwell-time weighting is determined by the jump weighting
through $\Delta r_i(u)$. Freeing it gives jumps and residence times independent
conjugate parameters,
\begin{equation}
M_{u,V}(\Gamma)=\exp\!\left[\tfrac12\sum_e u_e n_e-\sum_i\tau_i\big(\Delta r_i(u)+V_i\big)\right],
\label{eq:ext}
\end{equation}
and Proposition~\ref{prop:fk} survives with $g_{kl}\to g_{kl}-V_k-V_l$. Note that
$\Ex_0[M_{u,V}]\neq1$, so these are not likelihood ratios of any probability
measure; by Remark~\ref{rem:testfunctions} that does not matter. Because the
score acquires dwell-time components $s_{V,i}=-(\tau_i-\langle\tau_i\rangle)$,
the blind hyperplane must be recomputed in the full parameter space: it is
codimension one in $\mathbb R^{|E|+|V|}$ rather than in $\mathbb R^{|E|}$, and
that extra room is where much of the gain comes from.

\subsection{A two-channel machine with no shared states}\label{sec:twochannel}
Two independent reaction channels, each obeying local detailed balance: a slip
channel consuming $\Delta\mu$ and producing nothing, and a coupled channel
consuming $\Delta\mu$ and producing $w$. Four directed counters, each Poisson,
so $\eta=w\,n_c/[\Delta\mu(n_s+n_c)]$.

This is a deliberately different test. The cycles do not share a state space and
the counters are independent, against the chain's shared occupancy. Most
importantly, \emph{there is no state space, hence no residence times, hence no
site potentials}: the enlargement of Section~\ref{sec:sitepot} is structurally unavailable.
What plays its role is the symmetric part of the counter tilts, so ``two
antisymmetric parameters'' versus ``four independent'' is the two-channel
counterpart of ``edge tilts'' versus ``edge tilts plus site potentials''.

\begin{remark}[Matching the regime, not the parameters]
\label{rem:regime}
The model's original parameters use $\mathcal T=200$, which puts of order a
thousand events in a window. There, $\langle\Sigma\rangle=2681$; $\chi^{2}$ at
$|u|\sim0.25$ is about $e^{12.7}\approx3\times10^{5}$, so that even
$\Ex_0[M_u]=1$ cannot be estimated from $4\times10^{5}$ samples (we measured
$0.96$, $0.16$, $0.90$, $1.01$); and Cram\'er--Rao already reaches $99.9\%$ of
$\Var(\eta)$. Comparing like with like against the chain means matching events
per window, not nominal parameters. Section~\ref{sec:sweep} runs at $\mathcal T\in[1,8]$ and
keeps the long-window case deliberately.
\end{remark}

\section{Results}\label{sec:results}
All bounds below are computed by a split-sample protocol: coefficients are fitted
on the first half of the trajectories and the bound is evaluated on the second
half as a single-direction bound $\Cov(\eta,c\!\cdot\!M)^{2}/\Var(c\!\cdot\!M)$,
which is valid for any fixed $c$ by Cauchy--Schwarz and therefore cannot be
inflated by the fitting nor exceed $\Var(\eta)$.

\subsection{Blind directions, and what the test points buy}\label{sec:blindres}
\begin{figure}[H]
\centering
\includegraphics[width=\textwidth]{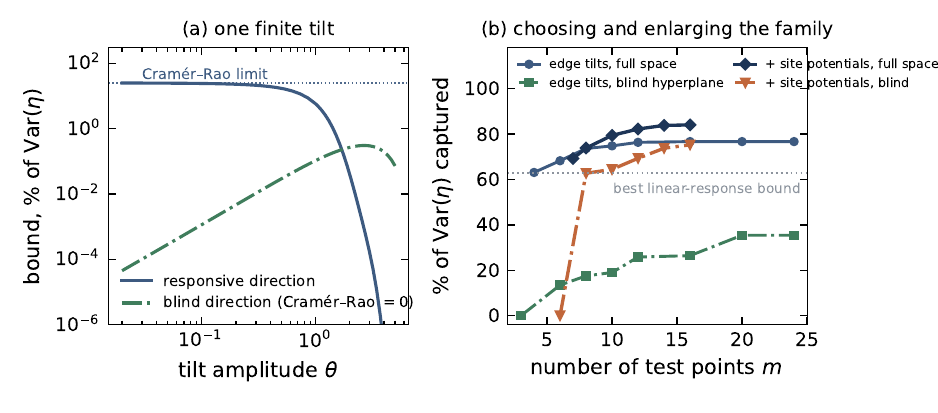}
\caption{(a) One finite perturbation of the three-state model. Along a
responsive direction the bound sits at its Cram\'er--Rao value (dotted) for small
$\theta$ and collapses beyond $\theta\sim1$. Along a blind direction
Cram\'er--Rao is exactly zero; the bound rises as $\theta^{2}$ from the origin,
as Proposition~\ref{prop:blind} requires, peaks near $\theta=2.5$ and falls.
(b) Several test points together. The horizontal line is the best
linear-response bound. Edge tilts alone reach $76.7\%$ in the full space and
$35.4\%$ on the blind hyperplane; adding site potentials raises these to $84.0\%$
and $75.2\%$}
\label{fig:blind}
\end{figure}

On the blind hyperplane the multiparameter Cram\'er--Rao bound is
$8\times10^{-26}$, and every design used satisfies
$\max_u|v\cdot u|\le3\times10^{-16}$, so the constraint is exact rather than
approximate. The bound there climbs from $0.297\%$ for one point to $15.4\%$ for
24 points chosen greedily from a random pool, $35.4\%$ with matching pursuit, and
$75.2\%$ once site potentials are admitted.

\subsection{Both models, across parameters}\label{sec:sweep}
\begin{table}[H]
\centering
\caption{Three-state model, extended family, $m=10$, best of the three selection
rules. CR is the multiparameter Cram\'er--Rao bound over edge tilts, i.e. the
best linear-response bound; ``gain'' is the ratio of the last column to it}
\label{tab:chain}
\begin{tabular}{lrrrrl}
\toprule
case & CR (\%) & CR extended (\%) & best (\%) & gain & winning rule\\
\midrule
baseline & 62.6 & 68.9 & 86.0 & $1.37\times$ & joint\\
strong drive & 73.2 & 78.6 & 94.4 & $1.29\times$ & dictionary\\
weak drive & 58.3 & 64.4 & 79.1 & $1.36\times$ & dictionary\\
near stall & 53.6 & 58.0 & 84.2 & $1.57\times$ & dictionary\\
fast futile cycle & 59.7 & 66.4 & 82.9 & $1.39\times$ & joint\\
short window & 61.2 & 72.8 & 93.8 & $1.53\times$ & joint\\
long window & 64.7 & 66.6 & 81.0 & $1.25\times$ & greedy\\
\bottomrule
\end{tabular}
\end{table}

\begin{table}[H]
\centering
\caption{Two-channel model, $m=10$, best of the three selection rules. CR$_2$ is
the multiparameter Cram\'er--Rao bound over the two antisymmetric directions,
CR$_4$ over all four counters. The last column is the bound on the blind
hyperplane, where both Cram\'er--Rao bounds are zero}
\label{tab:tc}
\begin{tabular}{lrrrrlr}
\toprule
case & CR$_2$ (\%) & CR$_4$ (\%) & best (\%) & gain & & blind (\%)\\
\midrule
very short, $\mathcal T=1$ & 66.8 & 68.1 & 92.0 & $1.38\times$ & & 77.5\\
baseline, $\mathcal T=2$ & 74.6 & 75.3 & 95.8 & $1.28\times$ & & 77.9\\
weak drive, $\mathcal T=2$ & 61.3 & 61.8 & 79.8 & $1.30\times$ & & 38.3\\
near stall, $\mathcal T=2$ & 52.9 & 53.4 & 72.3 & $1.37\times$ & & 50.3\\
strong drive, $\mathcal T=2$ & 97.9 & 97.9 & 99.9 & $1.02\times$ & & 99.1\\
medium, $\mathcal T=8$ & 96.6 & 96.6 & 99.7 & $1.03\times$ & & 54.4\\
long, $\mathcal T=60$ & 99.7 & 99.7 & 99.8 & $1.00\times$ & & \textbf{0.0}\\
\bottomrule
\end{tabular}
\end{table}

The two models agree where they should: gains of $1.25$--$1.57$ on the chain
against $1.28$--$1.38$ on the two-channel model in its comparable regime, and
blind recovery of $75.2\%$ against $77.9\%$, with Cram\'er--Rao zero to machine
precision in both.

\subsection{Where it stops working}\label{sec:limits}
They also disagree in two ways, and both are more informative than the agreement.

\emph{The gain is a short-window, far-from-Gaussian effect.} Read
Table~\ref{tab:tc} down the window column, together with Fig.~\ref{fig:tc}. As $\mathcal T$
grows the counts become Gaussian, the score becomes linear, and the projection
onto it captures essentially everything. Cram\'er--Rao climbs from $66.8\%$ to
$99.7\%$, the gain falls from $1.38\times$ to $1.00\times$, and the blind
hyperplane yields \emph{exactly nothing}. This is not a peculiarity of one model:
it explains why the chain's largest gains are at short window and near stall and
its smallest at long window (Table~\ref{tab:chain}), and it means every gain
quoted for this construction must carry its regime.

\begin{remark}[The scaling that keeps the divergence alive]
\label{rem:lan}
The collapse is a statement about \emph{fixed} perturbation amplitude. Since
$\chi^{2}+1$ is an exponential of a time integral, at fixed $u$ it grows
exponentially in $\mathcal T$; keeping it finite requires the diffusion scaling
$|u|\sim\mathcal O(\mathcal T^{-1/2})$ familiar from local asymptotic normality.
Measured on the three-state model along a fixed direction: at $|u|=0.3$,
$\chi^{2}$ runs $1.5\times10^{-2},\,6.0\times10^{-2},\,0.25,\,1.40,\,37.6$ for
$\mathcal T=2,8,30,120,500$, while at $|u|=0.85\,\mathcal T^{-1/2}$ it runs
$0.054,\,0.061,\,0.065,\,0.068,\,0.069$, effectively constant. But in that
scaled regime the perturbation is infinitesimal in the relevant sense and the
bound returns to Cram\'er--Rao. Finite-tilt projections are therefore a tool for
mesoscopic records, short enough that the central limit theorem has not yet
erased the non-Gaussianity they exploit.
\end{remark}

\begin{figure}[H]
\centering
\includegraphics[width=\textwidth]{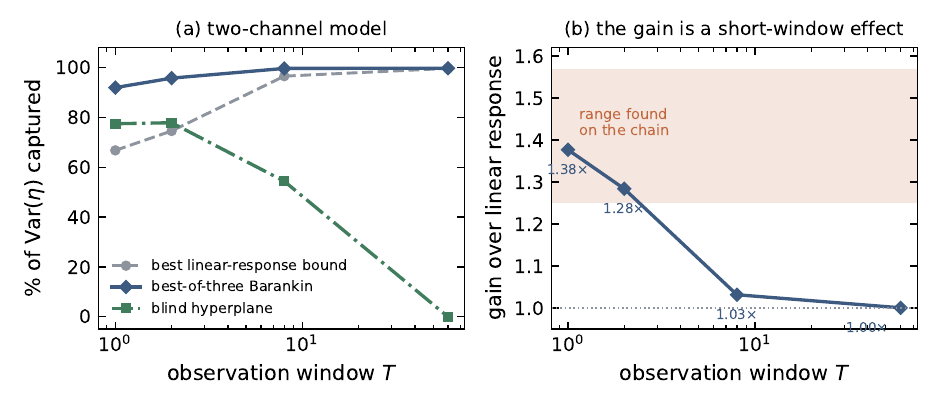}
\caption{(a) Two-channel model against observation window. As the window grows
the counts become Gaussian, the best linear-response bound approaches
$\Var(\eta)$, and the blind hyperplane (where that bound is identically zero)
yields exactly nothing. (b) The same as a ratio. The shaded band is the range
of gains found on the three-state model}
\label{fig:tc}
\end{figure}

\emph{The site-potential enlargement does not transfer.} On the chain it is the
largest single improvement in this work: $62.6\%\to68.9\%$ in Cram\'er--Rao and
$76.7\%\to84.0\%$ in the bound. On the two-channel model the analogous
enlargement, from two antisymmetric directions to four independent counters, is
worth $74.6\%\to75.3\%$ in Cram\'er--Rao and nothing in the bound, which is
$95.8\%$ (Table~\ref{tab:tc}). A separate run of the same three selection rules,
comparing the two sets of directions directly, gives $95.87\%$ for both; the
difference from the table is the run-to-run variation of the optimisers, about
$0.1$ percentage points.

The obvious explanation, that the extra directions fail to couple to $\eta$, is
wrong and in fact points the wrong way: the symmetric directions of the
two-channel model couple to $\eta$ nearly three times more strongly, relative to
the existing ones, than the site potentials of the chain, and deliver the
smaller gain (Appendix~\ref{sec:diagnostic}). What the extra directions must
supply is span that finite tilts of the existing directions cannot already
reach, and we know of no cheap test for it in advance.

\section{What coarse data cannot certify, and what counted data can}\label{sec:measure}

Sections~\ref{sec:design}--\ref{sec:results} measure what the construction recovers when
the model is known. This section asks what survives when it is not, and gives
the answer in two parts: a negative one for coarse data and a positive one for
edge-resolved records.

\subsection{Coarse data bound $\chi^{2}$ only from below}\label{sec:inference}
Every percentage reported above is obtained with test points chosen by
optimising against $\eta$ evaluated trajectory by trajectory. The split-sample
protocol of Appendix~\ref{sec:protocol} makes those numbers statistically sound, because the
coefficients never see the half on which the bound is scored, but it does not
make them useful to anyone who does not already hold the record. A record fine
enough to choose the design is a record fine enough to estimate $\Var(\eta)$
directly. We now ask what remains when that is not assumed.

Read at a single test point, the bound is

\begin{equation}
\Var(\eta)\;\ge\;\frac{\bigl[\langle\eta\rangle_u-\langle\eta\rangle\bigr]^{2}}
{\chi^{2}(u)} ,
\label{eq:oneshot}
\end{equation}

and the numerator is exactly what a coarse experiment can deliver: perturb the
machine by a known finite amount, measure the mean efficiency before and after.
No single trajectory need be resolved. The denominator is the difficulty, since
$\chi^{2}$ is a functional of the generator. For \eqref{eq:oneshot} to infer
anything, $\chi^{2}$ would have to be bounded \emph{above} by something
measurable.

Table~\ref{tab:oneshot} evaluates \eqref{eq:oneshot} on the three-state model
along a fixed direction, with $\chi^{2}$ computed exactly from
Proposition~\ref{prop:fk} and both means from $2\times10^{5}$ trajectories.

\begin{table}[H]
\centering
\caption{The measurable form \eqref{eq:oneshot} on the three-state model,
$\Delta\mu=4$, $w=2$, fixed tilt direction, $|u|=0.3$}
\label{tab:oneshot}
\begin{tabular}{rrrrrr}
\toprule
$\mathcal T$ & $\langle\eta\rangle$ & $\Delta\langle\eta\rangle$ & $\chi^{2}$ &
bound & \% of $\Var(\eta)$ \\
\midrule
 2 & 0.0938 & $-0.0080$ & $1.50\times10^{-2}$ & $4.31\times10^{-3}$ & 11.44 \\
 4 & 0.1366 & $-0.0077$ & $2.99\times10^{-2}$ & $1.96\times10^{-3}$ &  4.04 \\
 8 & 0.1940 & $-0.0075$ & $6.04\times10^{-2}$ & $9.21\times10^{-4}$ &  1.64 \\
16 & 0.2439 & $-0.0035$ & $1.24\times10^{-1}$ & $9.96\times10^{-5}$ &  0.18 \\
32 & 0.2732 & $-0.0024$ & $2.64\times10^{-1}$ & $2.23\times10^{-5}$ &  0.05 \\
\bottomrule
\end{tabular}
\end{table}

Eleven per cent at the shortest window, one part in two thousand by
$\mathcal T=32$. Optimising the direction would raise the first column somewhat,
but choosing that direction requires the generator, which is the obstruction
again rather than a way around it.

What matters in the table is the pattern. The measurable quantity
$\Delta\langle\eta\rangle$ is a difference of time-averaged means and changes
only slowly with $\mathcal T$ (by a factor of about three across the table),
while $\chi^{2}$ grows by more than an order of magnitude. The same two laboratory readings are therefore
consistent with bounds differing by orders of magnitude, according to a model
quantity the experiment cannot see. This is not an accident of this model or
this direction. Inequality \eqref{eq:oneshot}, read the other way, says that
measurable data bound $\chi^{2}$ from \emph{below}; that is what the
Hammersley--Chapman--Robbins construction is. Nothing in the data bounds it from
above, and a lower bound on a divergence cannot be inverted into the upper bound
\eqref{eq:oneshot} needs.

We therefore state the scope plainly. What is developed here is a computational
instrument for a model in hand: given a generator, it measures how much of a
variance finite-tilt projections recover, and exhibits directions along which
every linear-response bound is identically zero while this one is not. It is
not an uncertainty relation in the sense that makes those relations useful, and
Section~\ref{sec:dissipation} already showed why the obvious repair (paying for the perturbation in
dissipation) is unavailable. A reader looking for a bound that converts coarse
measurements into a statement about fluctuations will not find one here.

\subsection{Counted data: what an experiment would need}\label{sec:experiment}

Stochastic efficiencies have been measured in a colloidal Carnot
engine~\cite{Martinez2016}, but the bound needs more than a record of energies.
The construction requires $M_u$ per trajectory, and \eqref{eq:M} needs the
individual rates through $\Delta r_i(u)$. Those are countable: the
maximum-likelihood estimate of $w_{ij}$ from a trajectory record is
$\hat w_{ij}=N_{ij}/\mathcal{T}_i$, the number of $i\to j$ transitions over the
total time spent in $i$, both of which an edge-resolved single-molecule record
reports.

\begin{proposition}[Consistency of the counted bound]
\label{prop:count}
Let $\{\Gamma_n\}_{n=1}^{N}$ be i.i.d. steady-state trajectories of duration
$\mathcal T$ and let $\hat M_u$ be \eqref{eq:M} with $w$ replaced by $\hat w$.
Then $\hat G_{kl}\to G_{kl}$ almost surely.
\end{proposition}

\begin{proof}
By the strong law for Markov additive functionals,
$(N\mathcal T)^{-1}\mathcal T_i\to p_{0,i}$ and
$(N\mathcal T)^{-1}N_{ij}\to p_{0,i}w_{ij}$, so $\hat w_{ij}\to w_{ij}$ almost
surely, uniformly over the finite edge set. The exponent of \eqref{eq:M} is
\emph{linear} in $w$ with coefficients $\tau_i(\Gamma)\le\mathcal T$, so
\[
\big|\log\hat M_u(\Gamma)-\log M_u(\Gamma)\big|
\;\le\;\mathcal T\sum_{ij}\big|\hat w_{ij}-w_{ij}\big|\,\big|e^{u_{ij}/2}-1\big|
\]
uniformly in $\Gamma$; hence $\hat M/M\to1$ uniformly, not merely pointwise. The
sample average $N^{-1}\sum_n M_{u_k}M_{u_l}$ converges by the strong law, and the
difference between it and the plug-in version is bounded by
$\big(\sup|\hat M/M-1|\big)$ times it, which vanishes.
\end{proof}

Two caveats. $\Ex_0[\hat M_u]\neq1$ at finite $N$, so $\hat M_u$ is not a
likelihood ratio; by Remark~\ref{rem:testfunctions} the resulting bound is
nevertheless valid. And the same trajectories estimate $\hat w$ and evaluate the
bound, which sample-splitting removes.

Numerically the reconstruction is accurate well beyond what the rate errors
suggest. On the three-state model, with the slowest rate misestimated by $2.6\%$:

\begin{table}[H]
\centering
\caption{The bound computed from counted rates against the bound computed from
the true rates, three-state model, along the optimal direction}
\label{tab:count}
\begin{tabular}{lcccc}
\toprule
$|u|$ & $\chi^{2}$ (true) & $\chi^{2}$ (counted) & bound (true) & bound (counted)\\
\midrule
0.083 & 0.00206 & 0.00206 & $3.54146\times10^{-2}$ & $3.54118\times10^{-2}$\\
0.167 & 0.00812 & 0.00812 & $3.54837\times10^{-2}$ & $3.54810\times10^{-2}$\\
0.333 & 0.03178 & 0.03176 & $3.52121\times10^{-2}$ & $3.52097\times10^{-2}$\\
\bottomrule
\end{tabular}
\end{table}

So the denominator does not have to become thermodynamic in order to become
measurable; it has to become counted. For a rotary motor such as
$F_1$-ATPase the accounting closes without a kinetic model: the chemical
potential difference is fixed by the buffer, the hydrolysis number is counted,
and the work is read from the bead, so that $T_0S=n\Delta\mu-W$ realisation by
realisation~\cite{Toyabe2011}. The edge-resolved counts required here are of the
same character. For kinesin, Ariga, Tomishige and Mizuno report that the
measured heat and work together fall short of the input free-energy
change~\cite{Ariga2018}, a reminder that closing such balances is hard in
practice.

\section{An exactly solvable diffusion}\label{sec:diffusion}
The models of Section~\ref{sec:models} are finite-state, and every number in Section~\ref{sec:results} is
numerical. For a diffusion the construction can be carried through exactly in
one case, and three of the numerical observations above then become formulas:
the diffusion scaling of Remark~\ref{rem:lan}, the short-window optimum of
Section~\ref{sec:limits}, and the saturation of test-point placement studied in Appendix~\ref{sec:placement}.
The diffusion also brings in a constraint that finite-state models do not have:
a drift perturbation must keep the process confined.

\subsection{The Gram matrix on diffusion path space}\label{sec:ougram}
For a diffusion $dx=b(x)\,dt+\sqrt{2D}\,dW$ with constant $D$ and a drift tilt
$b\to b+u$, Girsanov gives
$M_u=\exp\big[\int u^{\top}(2D)^{-1/2}dW-\tfrac14\int u^{\top}D^{-1}u\,dt\big]$,
and the proof of Proposition~\ref{prop:fk} goes through line by line:
\begin{equation}
G_{kl}+1=\Ex_{u_k+u_l}\!\left[\exp\!\int_0^{\mathcal T}g_{kl}(x_t)\,dt\right],
\qquad g_{kl}=\tfrac12\,u_k^{\top}D^{-1}u_l .
\label{eq:diffusion}
\end{equation}
The potential is the Cameron--Martin cross term, and it is the small-tilt limit
of the jump potential of \eqref{eq:fk}. On a one-dimensional lattice of spacing
$h$ approximating the diffusion, the rates are of order $D/h^{2}$ and a drift
tilt enters as $u_{ij}/2=\pm u\,h/(2D)$, so
$\sum_j w_{ij}(e^{u_{k,ij}/2}-1)(e^{u_{l,ij}/2}-1)\to u_ku_l/(2D)$ as $h\to0$.
The Gram matrix becomes a Feynman--Kac semigroup of a differential operator
rather than a matrix exponential, and Proposition~\ref{prop:blind} carries over
formally, since $\chi_1$ is still linear in $u$. Kinetic tilts have no
diffusive counterpart with finite $\chi^{2}$: changing $D$ changes the quadratic
variation, and the two path measures are mutually singular. We checked
\eqref{eq:diffusion} against $1.5\times10^{5}$ Euler--Maruyama paths of an
Ornstein--Uhlenbeck process, for constant fields and for the space-dependent
pair $u=ax$, $v=c\sin x$; the relative deviations, $4\times10^{-4}$ to
$1.5\times10^{-2}$, are two to three times the Monte Carlo standard error and do
not shrink with the time step.

\subsection{The Ornstein--Uhlenbeck process and a quadratic observable}\label{sec:ou}
Take $dX=-\gamma X\,dt+\sqrt{2D}\,dW$ started from its stationary law
$\mathcal N(0,D/\gamma)$, the observable
$O=\mathcal T^{-1}\int_0^{\mathcal T}X_t^{2}\,dt$, and write $y=\gamma\mathcal T$.
A constant force $u$ is a blind direction: $O$ is even under $x\to-x$ and the
perturbation is odd, so the first-order response vanishes. The response is in
fact exactly quadratic,
\begin{equation}
\Delta\langle O\rangle(u)=\frac{u^{2}}{\gamma^{2}}R(y),
\qquad
R(y)=1-\frac{2}{y}\big(1-e^{-y}\big)+\frac{1}{2y}\big(1-e^{-2y}\big),
\label{eq:ouR}
\end{equation}
where $R$ is the transient cost of switching the force on at $t=0$. The
divergence and the variance are also exact:
\begin{equation}
\chi^{2}(u)=e^{\sigma^{2}u^{2}}-1,\quad \sigma^{2}=\frac{\mathcal T}{2D},
\qquad
\Var_0(O)=\frac{2D^{2}}{\gamma^{4}\mathcal T^{2}}\,\Psi(y),\quad
\Psi(y)=y-\tfrac12\big(1-e^{-2y}\big).
\label{eq:ouchi}
\end{equation}

\subsection{The optimal tilt}\label{sec:outilt}
With $c=R(y)/\gamma^{2}$ and $z=\sigma^{2}u^{2}$ the single-point bound is
\begin{equation}
\mathcal B_1(u)=\frac{[\Delta\langle O\rangle]^{2}}{\chi^{2}(u)}
=\frac{c^{2}}{\sigma^{4}}\,\phi(z),
\qquad \phi(z)=\frac{z^{2}}{e^{z}-1},
\label{eq:ouB1}
\end{equation}
maximised at the root of $e^{z}(2-z)=2$,
$z^{*}=2+\mathrm W_0(-2e^{-2})=1.593624$~\cite{Corless1996}, where
$\phi(z^{*})=0.647610$. The optimal amplitude is therefore
\begin{equation}
u^{*}=\sqrt{2Dz^{*}/\mathcal T}\;\propto\;\mathcal T^{-1/2},
\label{eq:oustar}
\end{equation}
which is the scaling measured numerically in Remark~\ref{rem:lan}, here obtained
as an exact optimum. The optimum is interior only because the direction is
blind. Along a direction that linear response sees, the numerator of
\eqref{eq:ouB1} is quadratic in $u$ rather than quartic, the ratio becomes
proportional to $z/(e^{z}-1)$, which decreases monotonically, and the best
single point is the limit $u\to0$, i.e.\ Cram\'er--Rao.

\subsection{The supremum, and what it is}\label{sec:ousup}
Along the constant-force axis the path enters only through one statistic:
$M_u=\exp(uY-\tfrac12u^{2}\sigma^{2})$ with $Y=W_{\mathcal T}/\sqrt{2D}\sim
\mathcal N(0,\sigma^{2})$, which is the likelihood ratio of
$\mathcal N(u\sigma^{2},\sigma^{2})$ against $\mathcal N(0,\sigma^{2})$. The
family is therefore a Gaussian location model in disguise, with Gram kernel
$G(u,v)=e^{\sigma^{2}uv}-1$ and response $b(u)=cu^{2}$. In the reproducing-kernel
space of $G$ the monomials are orthogonal with
$\|u^{n}\|^{2}=n!\,\sigma^{-2n}$~\cite{Duttweiler1973}, and $b$ is the single mode
$n=2$, so the Barankin supremum is
\begin{equation}
\mathcal B_\infty=\frac{2c^{2}}{\sigma^{4}},
\qquad
\frac{\mathcal B_\infty}{\mathcal B_1(u^{*})}=\frac{2}{\phi(z^{*})}=3.088277 .
\label{eq:ousup}
\end{equation}
$\mathcal B_\infty$ is the variance at $u=0$ of the unbiased estimator
$c\,(Y^{2}/\sigma^{4}-1/\sigma^{2})$ of $cu^{2}$, as Barankin's theorem says it
must be, and the ratio $3.088$ is the classical gap between the
Chapman--Robbins and Barankin bounds for the square of a Gaussian mean. We claim
no novelty for either constant, only that they appear here exactly. What the
example settles is the placement question of Section~\ref{sec:design}. The best of $300$ random
designs reaches $32.4\%$ of $\mathcal B_\infty$ with one point, $99.88\%$ with
two and $99.997\%$ with three: a single mode needs only enough points to resolve
it, whereas the edge-tilt family of Appendix~\ref{sec:placement} was still climbing at $m=20$.

As a fraction of the variance,
\begin{equation}
\mathcal Q_\infty(y)=\frac{\mathcal B_\infty}{\Var_0(O)}=\frac{4R(y)^{2}}{\Psi(y)},
\qquad
\mathcal Q_1(y)=\frac{\phi(z^{*})}{2}\,\mathcal Q_\infty(y).
\label{eq:ouQ}
\end{equation}

\begin{table}[H]
\centering
\caption{Recovered fraction of $\Var_0(O)$ for the Ornstein--Uhlenbeck process
along the constant-force axis, where the Cram\'er--Rao bound is identically
zero. Exact}
\label{tab:ouQ}
\begin{tabular}{rrr}
\toprule
$y=\gamma\mathcal T$ & $\mathcal Q_1$ & $\mathcal Q_\infty$ \\
\midrule
 0.50 &  2.39\% &  7.38\% \\
 1.00 &  6.45\% & 19.91\% \\
 2.00 & 12.44\% & 38.43\% \\
 3.61 & 14.96\% & 46.20\% \\
 8.00 & 11.40\% & 35.22\% \\
20.00 &  5.68\% & 17.55\% \\
\bottomrule
\end{tabular}
\end{table}

Both fractions vanish as $y^{3}$ at short windows and decay as $1/y$ at long
ones, with a maximum of $46.2\%$ at $y=3.606$ (Table~\ref{tab:ouQ}). This is the
analytic form of the short-window effect of Section~\ref{sec:limits}: the construction is
worth most when the window is a few relaxation times. Split-sample simulation
by the protocol of Appendix~\ref{sec:protocol} gives $13.33\%$ against a predicted $12.95\%$ for
$\mathcal Q_1$ at $y=2.14$, and $13.99\%$ against $13.66\%$ at $y=5.60$, the
residual being the Euler bias of the simulated response. The remaining $53.8\%$
of the variance is out of reach of the whole constant-force family, because $O$
is not a function of $Y$; reaching it requires space-dependent fields, and those
must confine.

Sampling at interval $\Delta t$ changes only the exponent: the exact discrete
likelihood gives $\chi^{2}_{\Delta}(u)=\exp[z\,\mathcal S(\gamma\Delta t)]-1$ with
$\mathcal S(x)=\tanh(x/2)/(x/2)$, so $z^{*}$ and the ratio \eqref{eq:ousup} are
unchanged and the optimal amplitude acquires the factor $\mathcal S^{-1/2}$. This
agrees with $4\times10^{5}$ sampled paths per seed, eight seeds, within one
standard error at $\gamma\Delta t=0.07$, $0.28$ and $0.98$.

\subsection{Confinement}\label{sec:confine}
On a finite state space no rate tilt can destabilise the process. A drift field
can: if it fails to confine, the perturbed process explodes and the stationary
law around which the response is expanded does not exist. For a linear
background drift $-A x$ the perturbed process is non-explosive and ergodic under
the drift condition~\cite{MeynTweedie1993}
\begin{equation}
\limsup_{|x|\to\infty}\frac{x\cdot u(x)}{|x|^{2}}<\gamma_{\min}(A),
\label{eq:cone}
\end{equation}
and the fields satisfying it form a convex cone, not a subspace. It excludes
the natural candidates: a Hermite mode $u\propto\psi_n$ grows as $x^{n}$, so
$xu(x)\to+\infty$ along one ray for every $n\ge2$. An admissible blind direction
does exist. For $u(x)=c_1x-c_3x^{3}$, with $\sigma_0^{2}=D/\gamma$, the
finite-window susceptibility is
\begin{equation}
\chi_1(\mathcal T)=\frac{\sigma_0^{2}}{\gamma}\,\big(c_1-3\sigma_0^{2}c_3\big)
\left[1-\frac{1-e^{-2\gamma\mathcal T}}{2\gamma\mathcal T}\right],
\label{eq:chi1T}
\end{equation}
which follows from It\^o's formula for $\Ex[X_t^{2}]$ at first order. It
vanishes at every window if and only if $c_1=3\sigma_0^{2}c_3$, i.e.\ for
$u\propto-(x^{3}-3\sigma_0^{2}x)$, a cubic Hermite polynomial. For $c_3>0$ this
field adds confinement; for $c_3<0$ the process explodes. Simulation at
$\gamma=1.3$, $D=0.6$ confirms both properties: over $1.2\times10^{5}$ paths
with $0<c_3\le0.3$ the process never leaves $|x|<3.6$, and finite differences
in the amplitude, extrapolated to zero, give $\chi_1$ within $0.4$ standard
errors of zero at $\mathcal T=1$ and $\mathcal T=4$.

A one-sided family costs little in principle and a great deal in practice. To
isolate the effect, restrict the constant-force family of Section~\ref{sec:ousup} to $u>0$,
as a cone would. The span of $\{G(\cdot,u):u\in\mathcal I\}$ is dense for any open
interval $\mathcal I\subset\mathbb R^{+}$ (the kernel sections are entire
functions, and a function orthogonal to all of them vanishes on a set with an
accumulation point), so the supremum is unchanged. But at $\sigma^{2}=c=1$, two
points placed symmetrically at $\pm u$ already reach $99.8\%$ of
$\mathcal B_\infty$ with $\kappa(G)=13$, because
$\tfrac12(M_u+M_{-u}-2)$ isolates the quadratic mode by parity, whereas equally
spaced points on $[0.4,2.4]$ need eight to reach the same fraction, at
$\kappa(G)=1.5\times10^{10}$. The ill-conditioning of Section~\ref{sec:trap} returns here
in exact form: an unregularised search over node positions on the ray returns
designs that exceed $\mathcal B_\infty$, which, as in Section~\ref{sec:trap}, is an
estimator failure and not a bound. The scripts for this section are part of
the replication package~\cite{FarihCode2026}.

\section{Discussion}\label{sec:discussion}
\subsection{What is established}\label{sec:established}
That an observable can be invisible to linear response and still be bounded. The
set of perturbation directions on which the first-order susceptibility vanishes
is a hyperplane, not a set of isolated accidents, and on it every bound built
from the Fisher information is identically zero. Proposition~\ref{prop:blind}
shows the second-order susceptibility takes over, and Sections~\ref{sec:design}--\ref{sec:results} show the
effect is large: three quarters of $\Var(\eta)$ on the chain, $77.9\%$ on the
two-channel model. We are not aware of a treatment of this situation (an
observable whose linear response vanishes on a whole family of perturbations,
requiring a finite-perturbation bound) anywhere in the stochastic
thermodynamics literature.

That efficiency is the natural example is not a coincidence, and
Proposition~\ref{prop:currents} says why. For any observable linear in the
trajectory record (every current, every residence time) the multiparameter
Cram\'er--Rao bound over rate and site perturbations already equals the
variance, so the question asked here is empty. It is non-empty exactly for
observables that are nonlinear in the record, and a ratio of two currents is
among the simplest such observables of thermodynamic interest. Its evenness under path
reversal, which places it outside the standard uncertainty relations, is a
second and independent reason.

The effect is not confined to three states. Embedded in a ring of up to sixteen
states, with a $33$-dimensional perturbation space, the same motor keeps a gain
of about $1.3$ over the best linear-response bound and about three quarters of
its variance on the blind hyperplane (Table~\ref{tab:largeN}).

\subsection{What is bounded}\label{sec:bounded}
Four limitations are structural rather than technical.

The construction is a short-window, far-from-Gaussian tool (Section~\ref{sec:limits}). It has
nothing to offer a machine observed long enough for its currents to become
Gaussian, and this should be stated wherever a gain is quoted.

It is not an inference tool (Section~\ref{sec:inference}). The designs that produce the
percentages reported here are fitted to $\eta$ evaluated trajectory by
trajectory, and a record that supports that fitting already supports estimating
$\Var(\eta)$ outright. The form an experiment could evaluate without
trajectory resolution requires an upper bound on $\chi^{2}$, which measurable
data do not supply: data bound a divergence from below, which is what the
Hammersley--Chapman--Robbins construction is in the first place. We regard this
as the sharpest limitation of the work, and we do not see a repair within the
present framework.

The cost is $\chi^{2}$, not dissipation, and Remark~\ref{rem:impossible} shows it
cannot be otherwise at finite perturbation. This is a real departure from the
thermodynamic uncertainty relation: what limits the achievable precision here is
the statistical distinguishability of the perturbed dynamics, and dissipation
enters only through the $\theta\to0$ limit.

Test-point placement is unsolved. We give three rules, none of which dominates
(Table~\ref{tab:select}). In the Gaussian example of Section~\ref{sec:diffusion} the question has
a complete answer, three points, but only because the response there is a single
mode of the kernel. Appendix~\ref{sec:placement} measures what that costs (about $1.2$
percentage points of standard deviation across random designs) and where the
edge-tilt family saturates, but the extrapolated ceiling there is a fit and not
a theorem, and we still do not know where the achievable bound saturates on the
blind hyperplane. We also have no criterion for when enlarging the family of test
functions will help: Section~\ref{sec:limits} records a plausible diagnostic that turns out to
point the wrong way. And our dictionary rule discretises Marzetta's continuum
formulation \eqref{eq:marzetta} rather than solving it; the integral equation for
the optimal weight function is the natural thing to race against matching
pursuit, and remains untried thirty years after it was posed.

\subsection{Open}\label{sec:open}
The most useful next step is a continuous-state model beyond the Gaussian case
of Section~\ref{sec:diffusion}: a nonlinear diffusion, where the Gram matrix
\eqref{eq:diffusion} has no closed form, or a non-equilibrium one with a
perturbation that is both blind and admissible. The counting reconstruction of
Section~\ref{sec:experiment} must there be replaced by estimates of $b$ and $D$ from sampled
increments, which is a statistical problem of a different kind. Whether blind
directions recover as large a share of the variance there as they do here is
the question we would most like answered.

A second question is quantitative. Barankin's theorem guarantees that the
supremum over test points is $\Var(O)$ exactly, so the gap between $75\%$ and
$100\%$ on the blind hyperplane is a statement about our search, not about the
family. A placement theory, or even a saturation rate in $m$, would convert
these numbers from measurements into predictions.

\appendix
\renewcommand{\theequation}{\thesection.\arabic{equation}}
\setcounter{equation}{0}
\renewcommand{\thetable}{\thesection.\arabic{table}}
\setcounter{table}{0}

\section{Numerical protocol and verification}\label{sec:protocolapp}
\subsection{Protocol}\label{sec:protocol}
The code implementing everything described in this appendix, together with the
transcripts of the runs reported in the paper, is deposited at
Ref.~\cite{FarihCode2026}.

Trajectories are generated by a vectorised Gillespie algorithm started from the
stationary distribution. Every bound reported in the paper is obtained as
follows. The trajectory ensemble is split in half. Test points are chosen and
coefficients $c$ fitted on the first half; the reported value is then evaluated
on the \emph{second} half as
\begin{equation}
\frac{\Cov(\eta,\,c\!\cdot\!M)^{2}}{\Var(c\!\cdot\!M)},
\label{eq:protocol}
\end{equation}
which by Cauchy--Schwarz is a valid lower bound on $\Var(\eta)$ for
\emph{any} fixed $c$. Two properties follow that make the numbers hard to
inflate: the reported value cannot exceed $\Var(\eta)$, and it cannot be raised
by overfitting the design, because the coefficients were never allowed to see
the data on which the value is computed.

Coefficients are obtained from a truncated eigendecomposition of $G$ after
standardising the columns to unit norm. The standardisation is not optional; see
Section~\ref{sec:trap}.

\subsection{Checks performed}\label{sec:checks}
\begin{table}[H]
\centering
\caption{Verification of the analytic results against independent computation}
\label{tab:checks}
\small
\begin{tabular}{p{0.42\textwidth}p{0.24\textwidth}p{0.26\textwidth}}
\toprule
claim & check & result\\
\midrule
Gram matrix \eqref{eq:fk}, diagonal & direct simulation, $N=4\times10^{5}$ &
relative error $\le10^{-4}$ for $\theta\in[0.05,1]$\\
Gram matrix \eqref{eq:fk}, off-diagonal & direct simulation &
$7\times10^{-3}$\\
$\chi^{2}(\theta)/\theta^{2}\to\mathcal I_F$ & $\theta=10^{-2},10^{-3},10^{-4}$ &
ratio $0.9939$, $0.9994$, $0.99994$\\
\eqref{eq:fk} on a one-state space & two-channel model, closed form vs.
simulation & $2.3\%$; and $\Ex_0[M_k]$ within $10^{-3}$ of $1$\\
Gradient \eqref{eq:grad} & central differences, $h\in[10^{-7},10^{-2}]$ &
$6\times10^{-9}$ relative, $0.0000^{\circ}$ angle\\
Blind constraint $u^{\top}v=0$ & every design used &
$\max|u^{\top}v|\le3\times10^{-16}$\\
Cram\'er--Rao on $\Hb$ & both models & $8\times10^{-26}$ (chain),
$1.3\times10^{-21}\%$ of $\Var(\eta)$ (two-channel)\\
Proposition~\ref{prop:currents} & three-state model, two currents and a
residence time & Cram\'er--Rao $=100.0000\%$ of the variance\\
Counted vs.\ true rates & Table~\ref{tab:count} & bound agrees to 4--5
significant figures at $2.6\%$ rate error\\
\bottomrule
\end{tabular}
\end{table}
\section{Supplementary numerical results}\label{sec:suppl}
\setcounter{table}{0}\setcounter{equation}{0}

\subsection{The point mass at $\eta=0$}\label{sec:pointmass}

Table~\ref{tab:pointmass} sizes the convention of Remark~\ref{rem:convention}.

\begin{table}[H]
\centering
\caption{Weight of the point mass at $\eta=0$, and its effect on the first two
moments. ``cond.'' denotes conditioning on $W+T_0S>0$, which is \emph{not} the
convention used in the paper and is shown only to size the effect}
\label{tab:pointmass}
\small
\begin{tabular}{lrrrrr}
\toprule
model & $\mathcal T$ & $\Pr[W{+}T_0S\le0]$ & $\langle\eta\rangle$ &
$\langle\eta\rangle$ cond. & $\Var(\eta)$ / $\Var(\eta)$ cond.\\
\midrule
three-state & 2 & 74.7\% & 0.093 & 0.366 & $3.73/4.73\times10^{-2}$\\
three-state & 4 & 60.5\% & 0.136 & 0.345 & $4.84/5.05\times10^{-2}$\\
three-state & 8 & 40.9\% & 0.193 & 0.327 & $5.62/5.13\times10^{-2}$\\
three-state & 20 & 14.7\% & 0.256 & 0.300 & $5.15/4.71\times10^{-2}$\\
two-channel & 1 & 2.66\% & 0.244 & 0.251 & $2.86/2.77\times10^{-2}$\\
two-channel & 2 & 0.16\% & 0.253 & 0.253 & $1.25/1.24\times10^{-2}$\\
\bottomrule
\end{tabular}
\end{table}

\subsection{Conditioning on disordered topologies}\label{sec:topo}

On the driven ring of Section~\ref{sec:scaling}, $\kappa(G)$ improves as the
network grows. The ring is translation invariant and its stationary
law is uniform, so its symmetry could be responsible for this behaviour. We
therefore repeated the computation on two families without any symmetry: random
3-regular graphs and open $L\times L$ grids. Each bond $e$ carries rates
$k_e e^{\pm F_e/2}$, with $\ln k_e$ drawn from $N(0,0.25)$ and $F_e$ uniform on
$[-2,2]$, so every network contains many cycles with different affinities and
its stationary law is far from uniform (the ratio of smallest to largest
stationary probability falls to $0.01$). Tilts are antisymmetric, one per bond,
so the tilt space has dimension $d=|E|$; all other settings are as for the
ring. Table~\ref{tab:topo} gives the median of $\kappa(G)$ at $m=12$ over five
independent draws of the network and the test points.

\begin{table}[H]
\centering
\caption{Conditioning of the Gram matrix at $m=12$ against the dimension $d$ of
the tilt space, for the driven ring and two disordered topologies. Medians over
five draws. The ring column is recomputed with the same protocol; it agrees in
order of magnitude with the single draws of Table~\ref{tab:scaling}}
\label{tab:topo}
\begin{tabular}{rlll}
\toprule
$d$ & ring & random 3-regular & $L\times L$ grid\\
\midrule
 6 & $1.3\times10^{3}$ & $1.5\times10^{3}$ & \\
12 & $1.7\times10^{2}$ & $2.1\times10^{2}$ & $1.8\times10^{2}$ ($L=3$)\\
24 & $1.8\times10^{1}$ & $4.2\times10^{1}$ & $2.9\times10^{1}$ ($L=4$)\\
40 &                   &                   & $7.5\times10^{0}$ ($L=5$)\\
48 & $5.1\times10^{0}$ & $8.8\times10^{0}$ & \\
\bottomrule
\end{tabular}
\end{table}

At equal $d$ the three topologies agree to within a factor of about two and a
half, with the disordered networks slightly worse conditioned than the ring, and
in all three $\kappa$ falls by a factor of $24$ to $34$ between $d=12$ and
$d\approx48$. The
improvement with network size is therefore set by the dimension of the tilt
space, not by the symmetry of the ring.

\subsection{Sensitivity to the design, and saturation}\label{sec:placement}

Section~\ref{sec:nodominate} shows that none of the three rules dominates, which raises two
questions: how much the answer depends on the design at all, and where adding
points stops paying. Neither has an analytical answer, since the Barankin
supremum over all designs is not computable, so we measure both.

One trajectory ensemble is generated and held fixed, so that the spread below is
variation over \emph{designs} and not Monte Carlo noise. Ten independent random
pools of forty tilts are drawn; points are added greedily from each and scored
on the held-out half by the protocol of Appendix~\ref{sec:protocol}. The basis always starts
with the four infinitesimal coordinate tilts, so $m=4$ reproduces multiparameter
Cram\'er--Rao ($62.7\%$ of $\Var(\eta)$ here) and everything above it is what
finite tilts add.

\begin{table}[H]
\centering
\caption{Bound as a percentage of $\Var(\eta)$ against the number of test
points, over ten independent random designs on the three-state model
($\mathcal T=8$, $N=2\times10^{5}$, edge-tilt family only)}
\label{tab:placement}
\begin{tabular}{rrrrr}
\toprule
$m$ & mean & s.d. & min & max \\
\midrule
 4 & 62.64 & 0.00 & 62.64 & 62.64 \\
 8 & 64.24 & 1.14 & 62.64 & 65.66 \\
12 & 67.05 & 0.97 & 65.62 & 68.71 \\
16 & 68.44 & 1.31 & 66.73 & 70.89 \\
20 & 69.35 & 1.21 & 67.71 & 71.31 \\
\bottomrule
\end{tabular}
\end{table}

Two things follow. The design matters, but at the scale of a few percentage
points and not tens: the standard deviation across pools settles near $1.2$
points and the full range never exceeds $4.3$. Whether our reported numbers come
from a lucky pool can be read off the min column: the worst of ten designs at
$m=20$ still beats Cram\'er--Rao by five points.

And the climb saturates. Fitting $f(m)=a-b\,e^{-cm}$ to the means gives
$a=74.1\%$, $b=15.0$, $c=0.060$, so the edge-tilt family is extrapolated to
reach about $74\%$ of $\Var(\eta)$ and no more; the fit puts $m=32$ at $71.9\%$
and $m=64$ at $73.7\%$, which is the returns-per-point argument against pushing
$m$ higher. We state this as a measured extrapolation and not a theorem. The
only rigorous ceiling remains $\Var(\eta)$ itself, and roughly a quarter of it is
not reached by this family at any $m$. That is a limitation of the family:
enlarging it with site potentials (Section~\ref{sec:sitepot}) moves the same model to $84.0\%$.

\subsection{A diagnostic that fails}\label{sec:diagnostic}

On the chain, site potentials are the largest single improvement in this work;
on the two-channel model the analogous enlargement is worth almost nothing
(Section~\ref{sec:limits}). The obvious explanation is that the extra directions fail to couple to $\eta$,
which would give a cheap \emph{a priori} test: measure the scores of the
candidate directions, and skip them if they are small. We tested it, and it is
wrong; in fact it points the wrong way. On the chain the site scores are
$\Cov_0(\eta,-\tau_i)=(-0.008,\,0.033,\,-0.025)$ against edge scores
$(0.021,\,0.128,\,0.065,\,0.018)$, a norm ratio of $0.29$, and the enlargement is
worth $+6.3$ points of Cram\'er--Rao. On the two-channel model the symmetric
directions score $(-0.042,\,0.046)$ against antisymmetric $(-0.044,\,0.067)$, a
ratio of $0.78$, nearly three times larger, and the enlargement is worth
$+0.7$ points. The direction with the \emph{larger} relative coupling delivers
the \emph{smaller} gain.

What the extra directions must supply is span that the existing directions
cannot already reach at finite tilt; coupling to $\eta$ is not enough. In the
two-channel model the symmetric directions are strongly correlated with $\eta$
and yet nearly redundant once finite antisymmetric tilts are admitted; on the
chain the dwell-time directions are weakly correlated and new. We know of no
cheap test that distinguishes these cases in advance and do not offer one: the
only procedure we can recommend is to enlarge the family and measure. We report
the failed diagnostic because it is the first thing a reader would try.

\section*{Data availability}
The data that support the findings of this article are openly
available~\cite{FarihCode2026}. The replication package contains all code
needed to reproduce every number, table and figure in the paper, together with
the outputs of the runs reported.

\section*{Acknowledgements}
I thank Prof. Abdessamad Ouchen for valuable discussions and explanations that helped shape this work. An earlier version of this manuscript was developed as two separate letters before being unified into the present article. 

I thank Jiming Zheng (University of North Carolina at Chapel Hill) for reading an earlier version of this manuscript and for pointing out its connections to Refs. [37–39]

AI tools were employed to assist with language editing, stylistic refinement, and code optimization.

\end{document}